\documentclass{article} 
\usepackage{authblk}
\usepackage{times}
\usepackage{fullpage}

\usepackage{hyperref, complexity}
\usepackage{mathtools}
\usepackage{bm, subcaption}
\usepackage{mdframed}
\usepackage{mathdots}
\usepackage{url,todonotes,xspace}
\usepackage{amsmath, amsthm}
\usepackage{amssymb}
\usepackage{tikz,xcolor}
\usepackage{xspace}
\usepackage{todonotes}
\usepackage{cleveref,environ}
\usepackage{mdframed}
\usepackage{xpatch}
\usepackage{floatrow}
\usepackage{thm-restate,tabularx}
\usepackage{colortbl,booktabs}
\tikzstyle{every picture} = [>=latex]
\usetikzlibrary{math, calc, arrows, quotes}
\usetikzlibrary {decorations.pathmorphing, decorations.pathreplacing, decorations.shapes, decorations.markings}
\usetikzlibrary{backgrounds}

\title{The Computational Complexity of Positive \\ Non-Clashing Teaching in Graphs\thanks{An extended abstract of this paper will appear in the proceedings of ICLR 2025.}}

\author[a]{Robert Ganian}
\author[a]{Liana Khazaliya}
\author[b]{Fionn~{Mc~Inerney}}
\author[a]{Mathis Rocton}

\affil[a]{Algorithms and Complexity Group, TU Wien, Austria}
\affil[b]{Telef\'{o}nica Scientific Research, Barcelona, Spain}

\date{}

\makeatletter
\newcommand{\probtitle}[1]{\gdef\@problemtitle{#1}}
\newcommand{\probinput}[1]{\gdef\@probleminput{#1}}
\newcommand{\probtask}[1]{\gdef\@problemtask{#1}}
\NewEnviron{problem}{
\probtitle{}\probinput{}\probtask{}
\BODY
\noindent
\begin{tabularx}{\textwidth}{
                     l X c}
   \multicolumn{2}{
                  l}{\@problemtitle} \\
   \textbf{Input:} & \@probleminput \\
   \textbf{Question:} & \@problemtask
\end{tabularx}
}
\makeatother

\surroundwithmdframed[linecolor=black,linewidth=1pt,skipabove=2mm,skipbelow=1mm,leftmargin=0mm,rightmargin=0mm,innerleftmargin=1mm,innerrightmargin=1mm,innertopmargin=1mm,innerbottommargin=1mm]{problem}

\newcommand{\bigoh}{\mathcal{O}}
\newcommand{\true}{\texttt{True}}
\newcommand{\false}{\texttt{False}}

\newcommand{\defquestion}[3]{
	\vspace{1mm}
	\noindent\fbox{
		\begin{minipage}{0.96\linewidth}
			\begin{tabular*}{\linewidth}{@{\extracolsep{\fill}}lr} \textsc{#1} & \\ \end{tabular*}
			{\bf{Input:}} #2 \\
			{\bf{Question:}} #3
		\end{minipage}
	}
	\vspace{1mm}
}

\newtheorem{theorem}{Theorem}
\newtheorem{proposition}[theorem]{Proposition}
\newtheorem{claim}[theorem]{Claim}
\newtheorem{definition}{Definition}
\newtheorem{observation}[theorem]{Observation}
\newtheorem{corollary}[theorem]{Corollary}
\newtheorem{lemma}[theorem]{Lemma}

\newcommand{\fvs}{\mathtt{fvs}}
\newcommand{\pw}{\mathtt{pw}}

\newcommand{\strictp}{\textsc{Strict Non-Clash}\xspace}
\newcommand{\normalp}{\textsc{Non-Clash}\xspace}
\newcommand{\bigo}{\mathcal{O}}

\newcommand{\eqB}{_{\sim_{\mathcal{B}}}}
\newcommand{\simB}{\sim_{\mathcal{B}}}
\newcommand{\equivBS}{_{\equiv_{\mathcal{B}}^{T}}}
\newcommand{\simBS}{\equiv_{\mathcal{B}}^{T}}
\newcommand{\core}{\mathcal{K}}
\newcommand{\coreBx}{\mathcal{K}_{(B,x)}}
\newcommand{\extcore}{\mathcal{K}_1}
\newcommand{\extcorebis}{\mathcal{K}_2}
\newcommand{\extcoreT}{\mathcal{K}^{T'}_2}
\newcommand{\cg}{\cellcolor[gray]{.9}}

\usetikzlibrary{arrows.meta, shapes, arrows, backgrounds, fit, positioning, patterns, calc}

\tikzset{
	circ/.style = {circle,draw,fill,inner sep=1.3pt}
}

\begin{document}

\maketitle

\begin{abstract}
We study the classical and parameterized complexity of computing the positive non-clashing teaching dimension of a set of concepts, that is, the smallest number of examples per concept required to successfully teach an intelligent learner under the considered, previously established model. For any class of concepts, it is known that this problem can be effortlessly transferred to the setting of balls in a graph~$G$. We establish (1) the \NP-hardness of the problem even when restricted to instances with positive non-clashing teaching dimension $k=2$ and where all balls in the graph are present, (2) near-tight running time upper and lower bounds for the problem on general graphs, (3) fixed-parameter tractability when parameterized by the vertex integrity of $G$, and (4) a lower bound excluding fixed-parameter tractability when parameterized by the feedback vertex number and pathwidth of $G$, even when combined with $k$.
Our results provide a nearly complete understanding of the complexity landscape of computing the positive non-clashing teaching dimension and answer open questions from the literature. 
\end{abstract}

\section{Introduction}
While typical machine learning models task a learner with finding a concept $C$ from a concept class~$\mathcal{C}$ based on an---often randomly drawn---sample, in \emph{machine teaching} (specifically in its commonly considered batch variant) the learner is provided a set of examples by a teacher; crucially, here the examples can be selected in a way which allows the concept to be reconstructed from as few examples as possible. Machine teaching is a core topic in computational learning theory and has found applications in a variety of areas including
robotics~\cite{TC09,ACYT12},
trustworthy AI~\cite{MZ15,ZZW18},
inverse reinforcement learning~\cite{HLMCA16,BN19},
and education~\cite{Zhu15,CAMPY18,ZSZR18}. 
While numerous models of machine teaching have been investigated to date~\cite{SM91,GK95,GM96,ZLH11,GCS17,MCV19,TelleHF19}, in this article we focus on the recently developed \emph{positive non-clashing teaching model}~\cite{KSZ19,FKS23}.

In non-clashing teaching, given a finite binary concept class $\mathcal{C}$, for each pair $C_1$, $C_2$ of distinct concepts in $\mathcal{C}$, at least one example provided for at least one of $C_1$ or $C_2$ must not be consistent with the other concept. A key feature of non-clashing teaching is that it is the most efficient model (in terms of the number of required examples) satisfying the Goldman-Mathias collusion-avoidance criterion~\cite{GM96}---the ``gold standard'' for ensuring that the learner cannot cheat, \emph{e.g.}, via a hidden communication channel with the teacher. Moreover, a teaching model is \emph{positive} if the examples provided for each concept $C$ are required to be positively labeled for $C$. The restriction of teaching models to the positive setting is common and well-motivated from applications in, \emph{e.g.}, recommendation systems~\cite{SPK00}, computational biology~\cite{WDM06,YJSS08}, and grammatical inference~\cite{SO94,Denis01}; see also the early works of Angluin~\cite{Angluininf,Angluinjcss}. Non-clashing teaching has been proposed and studied in the positive setting not only within the initial papers introducing the concept~\cite{KSZ19,FKS23}, but also in subsequent works (\emph{e.g.},~\cite{ChalopinCIR24}).

While positive non-clashing teaching has the potential to be highly efficient, realizing this potential requires us to solve the computational task of actually constructing a small set of examples for the given concepts. More precisely, one aims at computing (a witness for) the teaching dimension of a given concept class, \emph{i.e.}, the minimum integer $k$ such that each concept is provided with at most~$k$ examples that satisfy the conditions of the model. 
On the positive side, instead of considering various different types of concepts and examples, we can restrict our attention to the setting where each concept is a ball in some input-specified graph $G$, and the possible examples are vertices of~$G$. Indeed, it is known that any finite binary concept class $\mathcal{C}\subseteq 2^V$ can be represented by a set $\mathcal{B}$ of balls in a graph $G$ as follows: $V(G)=V\cup \{x_C \mid C\in \mathcal{C}\}$, $x_C$ is adjacent to $x_{C'}$ for all $C,C'\in \mathcal{C}$, $x_C$ is adjacent to $v\in V$ if and only if $v\in C$, and $\mathcal{B}=\{B_1(x_C)\mid C\in \mathcal{C}\}$~\cite{ChChMc,ChalopinCIR24}.

On the negative side, the problem is computationally intractable, and remains so even in the restricted setting where \emph{every} possible concept (\emph{i.e.}, every ball in $G$) is present; to distinguish this case from the general one where not all concepts need to be present, we refer to it as \emph{strict}. In particular, Chalopin, Chepoi, Mc~Inerney, and Ratel~\cite{ChalopinCIR24} recently carried out an initial complexity-theoretic investigation of computing the positive non-clashing teaching dimension in the strict setting. There, they established the \NP-hardness of the problem for instances with large teaching dimension (even when restricted to the highly restricted class of split graphs), obtained runtime upper and lower bounds under the \emph{Exponential Time Hypothesis}~\cite{IPZ}, and designed a so-called \emph{fixed-parameter} algorithm for the problem when parameterized by the size of the vertex cover of $G$. 

In this article, we significantly improve over each of these results, answer two open questions posed by the authors of the aforementioned work~\cite{ChalopinCIR24}, and obtain a nearly complete understanding of the computational complexity of computing the positive non-clashing teaching dimension (in both the strict and non-strict settings).

\paragraph*{Contributions.}
Let us refer to the problems of computing the positive non-clashing teaching dimension in the strict and non-strict settings as \strictp and \normalp, respectively. Formal definitions complementing the informal descriptions given above are provided in Section~\ref{sec:prelims}.

Our first result concerns the complexity of \strictp on instances with constant positive non-clashing teaching dimension. The reductions of Chalopin, Chepoi, Mc~Inerney, and Ratel~\cite{ChalopinCIR24} only establish the \NP-hardness of the problem for instances with large (\emph{i.e.}, input-dependent) positive non-clashing teaching dimension~$k$. In fact, as their first open question, the authors ask whether \strictp\ is \NP-hard or polynomial-time solvable when the sought-after dimension $k$ is a fixed constant; the question is not only theoretically interesting, but also highly relevant as instances of small teaching dimension are precisely the candidates for efficient teaching. We settle this by a highly non-trivial reduction (Theorem~\ref{thm:split2}) which establishes that determining whether the positive non-clashing teaching dimension is at most $2$ is \NP-hard---and remains so even when restricted to the same class of \emph{split graphs} where \strictp\ was previously shown to be \NP-hard (for large $k$)~\cite{ChalopinCIR24}. We note that this result is, in a sense, best possible: determining whether an instance of \strictp\ has a positive non-clashing teaching dimension of $1$ is trivial as it is equivalent to testing whether $G$ is edgeless~\cite{ChalopinCIR24}.

Next, we proceed to the running time bounds for solving the problem. Typically, the running time upper bounds are given by an exact algorithm, while the lower bounds are obtained from a suitable ``tight'' reduction under the Exponential Time Hypothesis~\cite{IPZ}. In the preceding work, Chalopin, Chepoi, Mc~Inerney, and Ratel~\cite{ChalopinCIR24} obtained algorithmic lower and upper bounds of $2^{o(n\cdot d)}$ and $2^{\bigoh(n^2\cdot d)}$, where $n$ and $d$ are the number of vertices and the diameter of $G$, respectively. From our reduction and a more careful algorithmic analysis, we obtain a lower bound of $2^{o(n\cdot d\cdot k)}$ and an upper bound of $2^{\bigoh(n\cdot d\cdot k\cdot \log n)}$ (Theorem~\ref{thm:ETHlower} and Proposition~\ref{pro:exactalgo})---making the bounds almost tight, with just a logarithmic factor in the exponent separating the two.

While the aforementioned bounds apply to the problem in general, often one may only need to solve the problem on ``well-structured'' graphs. The more refined \emph{parameterized complexity} paradigm~\cite{DowneyF13,CyganFKLMPPS15} offers the perfect tools to analyze and identify precisely which structural properties of input graphs---usually captured by a suitable integer \emph{parameter} $k$---allow us to circumvent its general intractability. The analog to the complexity class $\mathsf{P}$ in the parameterized setting is $\mathsf{FPT}$ (``\emph{fixed-parameter tractable}''), which characterizes parameterized problems solvable in $f(k)\cdot n^{\bigoh(1)}$ time; intuitively, this means that the problem is solvable in uniformly polynomial time for each constant value of $k$. Parameterized complexity is well-established and has been successfully applied for non-clashing teaching~\cite{ChalopinCIR24} as well as in a variety of related subfields of learning theory~\cite{DF93,LL18,GK21,OrdyniakS21,BGS23,EibenGKOS23,EibenOPS23}.

In their previous work, Chalopin, Chepoi, Mc~Inerney, and Ratel~\cite{ChalopinCIR24} established the fixed-parameter tractability of \strictp\ when parameterized by the \emph{vertex cover number} of the input graph $G$---or, equivalently, the vertex deletion distance to a graph consisting only of isolated vertices. The drawback of that result is that the vertex cover number is a highly ``restrictive'' parameter, in the sense that it achieves low values only on rather simple graphs. This is also reflected in the open question posed in that article, which asked about the problem's complexity under other parameterizations. As our third contribution, we establish (in Theorem~\ref{thm:fpt_vi}) the fixed-parameter tractability of \normalp---\emph{i.e.}, the more general task of computing the non-clashing teaching dimension when not all concepts (\emph{i.e.}, balls) need to be present---parameterized by the \emph{vertex integrity} of $G$. Vertex integrity is a well-studied graph parameter~\cite{DrangeDH16,GimaHKKO22,GimaO24,HanakaLVY24} that essentially captures the vertex deletion distance to a graph consisting only of small connected components; it is known (and easily observed) to be a less restrictive parameterization than the vertex cover number (cf. Section~\ref{sec:prelims}), meaning that our result significantly pushes the boundaries of tractability even for the simpler strict variant of the problem.

The proof of Theorem~\ref{thm:fpt_vi} is highly non-trivial. In particular, while the algorithm itself is simple and merely uses a data reduction technique (``\emph{kernelization}''~\cite{CyganFKLMPPS15}) that iteratively removes certain parts of the instance, the crucial correctness proof underlying the result is very involved and relies on identifying a carefully defined set of ``canonical'' examples for our instances. The algorithm is also constructive, meaning that it can output a set of examples for the concepts as a witness.

As our final contribution, in Theorem~\ref{thm:hard-fvs} we complement Theorem~\ref{thm:fpt_vi} with a complexity-theoretic lower bound excluding fixed-parameter tractability under many other graph parameters previously considered in the literature, including \emph{pathwidth}, \emph{treewidth}, and the \emph{feedback vertex number}---the latter two of which were explicitly mentioned in the aforementioned open question~\cite{ChalopinCIR24}. We do so through a complex \W[1]-\emph{hardness reduction} (which can be seen as a parameterized analog to classical reductions used to establish \NP-hardness) that excludes, under well-established complexity assumptions, \normalp\ from being in $\mathsf{FPT}$ even when combining all the parameterizations mentioned in the previous sentence with the positive non-clashing teaching dimension~$k$.

\paragraph*{Related Work.}
It is known that \normalp\ is significantly more challenging than the special case captured by \strictp. For instance, the reduction of Kirkpatrick, Simon, and Zilles~\cite[Subsection 7.1]{KSZ19} establishes that \normalp\ is \NP-hard even if the task is to determine whether the positive non-clashing teaching dimension of the instance is~$1$; on the other hand, the analogous question for \strictp\ is trivial as it simply requires determining whether the input graph is edgeless or not~\cite{ChalopinCIR24}. In fact, unless $\P=\NP$, that reduction also rules out a polynomial-time $1.999$-approximation algorithm. For clarity, note that while that reduction does not consider concepts that are balls in a graph, as mentioned earlier, \emph{every} finite binary concept class can be easily transformed into a class of balls in a graph~\cite{ChChMc,ChalopinCIR24}.

Apart from the computational questions resolved in this work, another prominent open question is whether the non-clashing teaching dimension is upper-bounded by the VC-dimension~\cite{KSZ19,FKS23,Si}. It is known that the non-clashing teaching dimension (where one allows negative examples) can be significantly smaller than the positive variant, {\it e.g.}, balls in cycles have a non-clashing teaching dimension of $2$, but their positive non-clashing teaching dimension is not bounded by any constant~\cite{ChalopinCIR24}. It was also pointed out that balls in cacti or planar graphs could be good candidates for concept classes negatively answering this question~\cite{ChalopinCIR24}. Further, Simon~\cite{Si} recently explored the relationship between non-clashing teaching and \emph{recursive teaching} and identified the precise gap between the two notions.

Concept classes consisting of balls in a graph are a discrete analog of the geometric concept classes of balls in a Euclidean space which have been investigated in PAC-learning, {\it e.g.}, as part of the more general Dudley concept classes~\cite{Fl,BDLi}. Apart from non-clashing teaching, they have also been explored for the closely related and well-studied sample compression schemes~\cite{ChChMc} introduced by Littlestone and Warmuth~\cite{LiWa}. As discussed in prior works~\cite{KSZ19,FKS23,ChalopinCIR24}, non-clashing teaching maps can be viewed as signed variants of representation maps for concept classes, a notion introduced to design unlabeled sample compression schemes for maximum concept classes~\cite{KuWa} (and subsequently the more general ample concept classes~\cite{ChChMoWa}).

\section{Preliminaries}
\label{sec:prelims}
We assume familiarity with graph terminology~\cite{Diestel}. We only consider simple, finite, and undirected graphs. For an integer $n \geq 1$, we set $[n] := \{1, \dots, n\}$. 
As we only consider finite binary concept classes which can be represented as balls in graphs~\cite{ChChMc,ChalopinCIR24}, we introduce the terminology for positive non-clashing teaching directly in the setting of graphs.

\subparagraph*{Positive Non-Clashing Teaching in Graphs.}
Let $G$ be a graph. For an integer $r\geq 0$ and a vertex $v\in V(G)$, the \emph{ball} $B_r(v)$ is the set of all vertices at distance at most $r$ from its \emph{center} $v$. Let $\mathcal{B}$ be a set of balls of $G$. A \emph{positive teaching map} $T$ for $\mathcal{B}$ is a mapping which assigns to each ball $B\in \mathcal{B}$ a \emph{teaching set} $T(B)\subseteq B$, \emph{i.e.}, a subset of the vertices of $B$. The \emph{dimension} of $T$ is $\max_{B\in \mathcal{B}}|T(B)|$---in other words, the largest image of $T$. A positive teaching map $T$ is called \emph{non-clashing} for $\mathcal{B}$ if for each pair of distinct balls $B_1, B_2\in \mathcal{B}$, there exists a vertex $w\in T(B_1)\cup T(B_2)$ such that $w\not \in B_1\cap B_2$. Note that~$w$ must lie in $B_1\cup B_2$ by definition, and hence, this condition ensures that one of the balls has a teaching set which is not contained in the other ball. We say that $w$ \emph{distinguishes} $B_1$ and $B_2$, or distinguishes $B_1$ from $B_2$ (or vice versa). If a teaching map is not non-clashing, we say that there is a \emph{conflict} between any two balls for which there is no element distinguishing them.
We now define our problems of interest:\footnote{While we use decision variants, our algorithms are constructive and can output a teaching map as a~witness.}

\begin{problem}
   \probtitle{\strictp}   
   \probinput{A graph $G$ and an integer $k$.}
   \probtask{Is there a positive non-clashing teaching map for the set of all balls of $G$ with dimension at most $k$?}
\end{problem}
\begin{problem}
   \probtitle{\normalp}   
   \probinput{A graph $G$, a set $\mathcal{B}$ of balls of $G$, and an integer $k$.}
   \probtask{Is there a positive non-clashing teaching map for $\mathcal{B}$ with dimension at most $k$?}
\end{problem}

We call a teaching map satisfying the conditions of the respective problem statement a \emph{solution}. To avoid any confusion, we remark that the above definitions---as well as every result obtained in this article---concerns non-clashing teaching in the previously studied \emph{positive} setting.

\subparagraph*{Parameterized Complexity.}
In parameterized
complexity~\cite{DowneyF13,CyganFKLMPPS15}, the
running-times of algorithms are studied with respect to a parameter
$p\in \mathbb{N}$ and input size~$n$. 
It is normally used for \NP-hard problems, with the aim of finding a parameter describing a feature of the instance such that the combinatorial explosion is confined to this parameter. 
A parameterized problem is \emph{fixed-parameter tractable} (\FPT) if it can be solved by an algorithm running in time $f(p)\cdot n^{\bigoh(1)}$, where $f$ is a computable function; these are \emph{fixed-parameter algorithms}.

Proving that a problem is $\W[1]$-hard via a \emph{parameterized reduction} from a $\W[1]$-hard problem $\mathcal{P}$ rules out the existence of a fixed-parameter algorithm under the well-established hypothesis that $\W[1]\neq \FPT$.
 A parameterized reduction from $\mathcal{P}$ to a parameterized problem $\mathcal{Q}$ is a function:
\begin{itemize}
\item which maps \textbf{YES}-instances to \textbf{YES}-instances and \textbf{NO}-instances to \textbf{NO}-instances,
\item is computable in time $f(p)\cdot
n^{\bigoh(1)}$, where $f$ is a computable function, and
\item where the parameter of the output instance can be upper-bounded by some function of the parameter of the input instance.
\end{itemize}

\strictp\ is known to be fixed-parameter tractable when parameterized by the \emph{vertex cover number} of $G$, \emph{i.e.}, the minimum integer $a$ such that 
there is a subset $X\subset V(G)$ of at most $a$ vertices where $G-X$ is an edgeless graph.
In this article, we consider three parameters which are upper-bounded by the vertex cover number (or, more precisely, the vertex cover number plus one):

\begin{itemize}
\item the \emph{vertex integrity} of $G$, which is the minimum integer $b$ such that there is a vertex subset $X\subset V(G)$ where for every connected component $H$ of $G-X$,  $|V(H)\cup X|\leq b$;
\item the \emph{feedback vertex number} of $G$ (denoted by $\fvs(G)$), which is the minimum integer $c$ such that there is a vertex subset $X\subset V(G)$ where $G-X$ is acyclic;
\item the \emph{pathwidth} of $G$ (denoted by $\pw(G)$), which has a more involved definition based on the notion of \emph{path decompositions}. However, for the purposes of this article it is sufficient to note the well-known facts~\cite{DowneyF13,CyganFKLMPPS15} that deleting one vertex from each connected component of $G$ will decrease the pathwidth by at most one, and that a graph consisting of a disjoint union of paths and \emph{subdivided caterpillars} (\emph{i.e.}, graphs consisting of a central path with pendent paths attached to it) has pathwidth~$2$.
\end{itemize}

\section{Intractability and Running Time Lower Bounds}
In this section, we establish the \NP-hardness of \strictp\ when $k=2$, and thus, that it cannot be $1.499$-approximated in polynomial time unless $\P=\NP$; in fact, our results hold even when the graphs belong to the class of \emph{split graphs}, \emph{i.e.}, graphs which can be partitioned into an independent set and a clique.
Recall that the former result is tight in the sense that \strictp\ is trivial when $k=1$ \cite{ChalopinCIR24}.
We formalize the result below.

\begin{theorem}\label{thm:split2}
	\strictp\ is \NP-hard even when restricted to split graphs with $k=2$.
\end{theorem}   

We prove Theorem~\ref{thm:split2} via a polynomial-time reduction that, given an instance of 3-\SAT, constructs an equivalent instance $(G,k)$ of \strictp, where $G$ is a split graph and $k=2$.

\defquestion{3-\SAT}{A CNF formula over a set of clauses $\mathcal{C}=\{c_1, \dots, c_{m}\}$ containing variables from $\mathcal{X}= \{x_1, \dots, x_n\}$, where each clause has exactly $3$ literals.}{Is there a variable assignment $\tau: \mathcal{X} \rightarrow \{\true, \false\}$ satisfying each clause in~$\mathcal{C}$?}

\noindent\textit{Reduction.}
	Given an instance $\phi=(\mathcal{C}, \mathcal{X})$ of 3-\SAT, we construct the graph~$G$ as follows (see Figure~\ref{fig:reduction_np} for an illustration).
	\begin{figure}[ht]
		\centering
		\subcaptionbox{The force-gadget. The dotted edge corresponds to the absence of that edge.\label{fig:gadget}}[.2\textwidth]{
		\scalebox{.78}{\begin{tikzpicture}

\node[circ, label=left:{$v'$}] (s'k) at (0, 1.5) {};
\node[circ, label=left:{$v''$}] (s''k) at (0, .5) {};
\node[circ, label=above:{$v^*$}] (s*k) at (1, 2) {};
\node[circ, label=above:{$v^{**}$}] (s**k) at (1, 1) {};
\node[circ, label=above:{$v^{***}$}] (s***k) at (1, 0) {};
\draw (s'k) -- (s*k);
\draw (s'k) -- (s**k);
\draw[dotted] (s'k) -- (s***k);
\draw (s''k) -- (s*k);
\draw (s''k) -- (s**k);
\draw (s''k) -- (s***k);
\end{tikzpicture}}}
		\hfill
		\subcaptionbox{An example of the graph $G$ obtained by applying our reduction on the 3-\SAT\ instance with $\mathcal{X} = \{x_i, x_j, x_q, x_p\}$ and $\mathcal{C} = \{(x_i\vee x_j\vee \overline{x_q}), (x_j\vee x_q\vee \overline{x_p})\}$. Vertices in ovals form independent sets, while cliques are depicted by rectangles. Blue edges denote the existence of all possible edges between the two sets. \label{fig:reduction_np}}[.75\textwidth]{
			\scalebox{.7}{\begin{tikzpicture}
\def\y{3}\def\x{-.5}

\node[circ, label=left:{$s'_k$}] (s'k) at (\x, \y+6.5) {};
\node[circ, label=left:{$s''_k$}] (s''k) at (\x, \y+5.5) {};
\node[circ, label=above:{$s^*_k$}] (s*k) at (1, \y+7) {};
\node[circ, label=above:{$s^{**}_k$}] (s**k) at (1, \y+6) {};
\node[circ, label=above:{$s^{***}_k$}] (s***k) at (1, \y+5) {};

\node[circ, label=left:{$s'_l$}] (s'l) at (\x, \y+1.5) {};
\node[circ, label=left:{$s''_l$}] (s''l) at (\x, \y+.5) {};
\node[circ, label=above:{$s^*_l$}] (s*l) at (1, \y+2) {};
\node[circ, label=above:{$s^{**}_l$}] (s**l) at (1, \y+1) {};
\node[circ, label=above:{$s^{***}_l$}] (s***l) at (1, \y+0) {};

\foreach \v in {k,l} {
	\draw (s'\v) -- (s*\v);
	\draw (s'\v) -- (s**\v);
	\draw (s''\v) -- (s*\v);
	\draw (s''\v) -- (s**\v);
	\draw (s''\v) -- (s***\v);
}

\node[circ, label=right:{$r'_i$}] (r'i) at (7.5, \y+7.5) {};
\node[circ, label=right:{$r''_i$}] (r''i) at (7.5, \y+6.5) {};
\node[circ, label=above:{$r^*_i$}] (r*i) at (6, \y+8) {};
\node[circ, label=above:{$r^{**}_i$}] (r**i) at (6, \y+7) {};
\node[circ, label=above:{$r^{***}_i$}] (r***i) at (6, \y+6) {};

\node[circ, label=right:{$r'_j$}] (r'j) at (7.5, \y+4.5) {};
\node[circ, label=right:{$r''_j$}] (r''j) at (7.5, \y+3.5) {};
\node[circ, label=above:{$r^*_j$}] (r*j) at (6, \y+5) {};
\node[circ, label=above:{$r^{**}_j$}] (r**j) at (6, \y+4) {};
\node[circ, label=above:{$r^{***}_j$}] (r***j) at (6, \y+3) {};

\node[circ, label=right:{$r'_q$}] (r'q) at (7.5, \y+1.5) {};
\node[circ, label=right:{$r''_q$}] (r''q) at (7.5, \y+.5) {};
\node[circ, label=above:{$r^*_q$}] (r*q) at (6, \y+2) {};
\node[circ, label=above:{$r^{**}_q$}] (r**q) at (6, \y+1) {};
\node[circ, label=above:{$r^{***}_q$}] (r***q) at (6, \y+0) {};

\node[circ, label=right:{$r'_p$}] (r'p) at (7.5, \y-1.5) {};
\node[circ, label=right:{$r''_p$}] (r''p) at (7.5, \y-2.5) {};
\node[circ, label=above:{$r^*_p$}] (r*p) at (6, \y-1) {};
\node[circ, label=above:{$r^{**}_p$}] (r**p) at (6, \y-2) {};
\node[circ, label=above:{$r^{***}_p$}] (r***p) at (6, \y-3) {};

\foreach \v in {i,j,q,p} {
	\draw (r'\v) -- (r*\v);
	\draw (r'\v) -- (r**\v);
	\draw (r''\v) -- (r*\v);
	\draw (r''\v) -- (r**\v);
	\draw (r''\v) -- (r***\v);
}

\node[circ, label=above:{$t_i$}] (ti) at (3.5, \y+7) {};
\node[circ, label=above:{$f_i$}] (fi) at (3.5, \y+6) {};
\node[circ, label=above:{$t_j$}] (tj) at (3.5, \y+4.5) {};
\node[circ, label=above:{$f_j$}] (fj) at (3.5, \y+3.5) {};
\node[circ, label=above:{$t_q$}] (tq) at (3.5, \y+2) {};
\node[circ, label=above:{$f_q$}] (fq) at (3.5, \y+1) {};
\node[circ, label=above:{$t_p$}] (tp) at (3.5, \y-.5) {};
\node[circ, label=above:{$f_p$}] (fp) at (3.5, \y-1.5) {};

\foreach \v in {i,j,q,p} {
	\draw (t\v) -- (r***\v);
	\draw (f\v) -- (r***\v);
	\draw (t\v) -- (r*\v);
	\draw (f\v) -- (r*\v);
}

\draw (s*k) -- (fi);
\draw (s***k) -- (fi);
\draw (s*k) -- (fj);
\draw (s***k) -- (fj);
\draw (s*k) -- (tq);
\draw (s***k) -- (tq);

\draw (s*l) -- (fq);
\draw (s***l) -- (fq);
\draw (s*l) -- (fj);
\draw (s***l) -- (fj);
\draw (s*l) -- (tp);
\draw (s***l) -- (tp);

\node[circ, label=left:{$s'_0$}] (s'0) at (\x, \y-.5) {};
\node[circ, label=right:{$r'_0$}] (r'0) at (7.5, \y-3.5) {};

\node[ellipse, color = gray, label=above:{$S'$}, draw, fit={(s''l) (s''k) (-.8,5)},inner sep=4mm] (S') {}; 
\node[ellipse, color = gray, label=above:{$R'$}, draw, fit={(7.5, \y-2) (r''i) (8,5)},inner sep=3mm] (R') {}; 
\node[ellipse, color = gray, label=above:{$A$}, draw, fit={(fi) (tp)},inner sep=5mm] (A) {};

\node [label=above:{$S^*$}] (S*) at (1, 6.75) {
	\tikz {\draw[color = gray] (0, 0) rectangle ++(1,8.25);}
};
\node [label=above:{$R^*$}] (R*) at (6, 5.5) {
	\tikz {\draw[color = gray] (0, 0) rectangle ++(1,12.25);}
};

\node[circ, label=left:{$a$}] (a) at (0, 12.5) {};

\draw [blue] ($(S'.south)$) to [bend right=30] ($(R*.south)+(0,0.12)$);
\draw [blue] ($(R'.south)$) to [bend left=43] ($(S*.south)+(0,0.12)$);
\draw [blue] ($(R*.south)+(0,0.12)$) to [bend left=40] ($(S*.south)+(0,0.12)$);

\draw [blue] (a) to [bend left=10] ($(R'.north)$);
\draw [blue] (a) to [bend left=15] ($(R*.north)-(0,0.12)$);
\draw [blue] (a) to [bend left=10] ($(A.north)$);
\draw [blue] (a) to [bend left=20] ($(S'.north)$);
\draw [blue] (a) to [bend right=20] ($(S*.north)-(0,0.12)$);
\end{tikzpicture}}}
	\end{figure}
	\begin{itemize}
		\item First, for each $i\in [n]$, we create a pair of vertices $\{t_i, f_i\}$. We set $A:=\{t_i,f_i\}_{i\in [n]}$.
		\item For each $i\in [n]$, we introduce a \emph{variable force-gadget}, which consists of a set of vertices $\{r_i^{*}, r_i^{**}, r_i^{***}, r_i', r_i''\}$ and edges as depicted in Figure~\ref{fig:gadget}.
		\item For each $i\in [n]$, we attach the variable force-gadget to the pair $\{t_i, f_i\}$ as shown in Figure~\ref{fig:reduction_np}, by making both $r_i^*$ and $r_i^{***}$ adjacent to both $t_i$ and $f_i$. This gadget will guarantee that the corresponding assignment of the $i^{\text{th}}$ variable is well-defined.
		We set $R^*:=\{r_i^{*}, r_i^{**}, r_i^{***}\}_{i\in [n]}$ and $R':=\{r_0'\}\cup\{r_i', r_i''\}_{i\in [n]}$, where $r_0'$ is a new vertex.
		\item Similarly, for each $k\in [m]$, we introduce a \emph{clause force-gadget} on the set of new vertices $\{s_k^{*}, s_k^{**}, s_k^{***}, s_k', s_k''\}$ (as depicted in Figure~\ref{fig:gadget}).
		This gadget corresponds to the clause $c_k$ of the instance $\phi$.
		We add adjacencies according to the appearance of literals in $c_k$, \emph{i.e.}, if $x_i\in c_k$ for some $i\in [n]$, then we connect both $s_k^{*}$ and $s_k^{***}$ to~$f_i$; and if $\overline{x_i}\in c_k$, then we connect both $s_k^{*}$ and $s_j^{***}$ to $t_i$. Intuitively, we connect the gadget to those vertices whose underlying assignments (\true\ or \false\ for $t_i$ and $f_i$, resp.) \textbf{do not} satisfy $c_j$, while the opposite assignments would (see Figure~\ref{fig:reduction_np}).
		We set $S^*:=\{s_k^{*}, s_k^{**}, s_k^{***}\}_{k\in [m]}$ and $S':=\{s_0'\}\cup\{s_k', s_k''\}_{k\in [m]}$, where $s_0'$ is a new vertex.
		\item We add all possible edges between (a)~$S^*$ and $R^*$; (b)~$R^*$ and $S'$; (c)~$S^*$ and $R'$.
		\item We add all possible edges within $S^*$, and within $R^*$.
		\item Lastly, we add a special vertex $a$ and make it adjacent to all the other vertices of the graph.
	\end{itemize}
This concludes the construction of $G$; given an instance $\phi=(\mathcal{C}, \mathcal{X})$ of 3-\SAT, the reduction outputs the \strictp instance $(G, 2)$. We now prove its correctness via the next two lemmas.

\begin{lemma}
	\label{lemma:correcntes1}
	If $\phi$ is a YES-instance of 3-\SAT, then $(G,2)$ is a YES-instance of \strictp.
\end{lemma}

\begin{proof}
	As $G$ contains a universal vertex $a$, its diameter is $2$, and thus, any ball of radius at least~$2$ contains $V(G)$. Hence, it is sufficient to consider balls of radius at most $2$. 	Let $\tau: \mathcal{X} \rightarrow \{\true, \false\}$ be an assignment of the variables satisfying the given 3-\SAT\ formula $\phi$. Let us define a positive teaching map $T$ of dimension $2$ for the set $\mathcal{B}$ of all the balls in $G$ as shown in Table~\ref{tab:map}. It remains to prove that $T$ is non-clashing for $\mathcal{B}$.
		\begin{table}[t]\renewcommand{\arraystretch}{1.2}
			\begin{tabular}{ l l l l }
				for $k\in [m]$ & \multicolumn{3}{l}{$T(B_1(s_k^*))=\{s'_k, v\}$ and $c_k$ is satisfied by $\tau(x_i)$ for some $i\in [n]$, where}\\
				& \multicolumn{3}{l}{\hspace{.9cm}if $\tau(x_i)=\true$, then $v=f_i$\hspace{1.2cm}($f_i\in B_1(s_k^*)\cap A$ since $x_i\in c_k$)}\\ 
				& \multicolumn{3}{l}{\hspace{.9cm}if $\tau(x_i)=\false$, then $v=t_i$\hspace{1.05cm}($t_i\in B_1(s_k^*)\cap A$ since $\overline{x_i}\in c_k$)}\\ 
				& $T(B_1(s_k^{**}))=\{s'_k, r_0'\}$ & $T(B_1(s_k^{***}))=\{s''_k, r'_0\}$& \\
				& $T(B_1(s_k'))=\{s_k', s_k^{**}\}$ & $T(B_1(s_k''))=\{s_k'', s_k^{***}\}$ &  $T(B_1(s'_0))=\{s'_0, a\}$\\\cline{2-4}
				for $i\in [n]$ & \multicolumn{3}{l}{$T(B_1(r_i^*))=\{r'_i, v\}$, where} \\
				& \multicolumn{3}{l}{\hspace{.9cm}if $\tau(x_i)=\true$, then $v=t_i$} \\
				& \multicolumn{3}{l}{\hspace{.9cm}if $\tau(f_i)=\false$, then $v=f_i$} \\
				& $T(B_1(r_i^{**}))=\{r'_i, s'_0\}$ & $T(B_1(r_i^{***}))=\{r''_i, s'_0\}$& \\
				& $T(B_1(r_i'))=\{r'_i, r^{**}_i\}$ & $T(B_1(r_i''))=\{r_i'', r^{
					***}_i\}$ & $T(B_1(r'_0))=\{r'_0, a\}$ \\\cline{2-4}
				for $i\in [n]$ & $T(B_1(t_i))=\{t_i, a\}$ & $T(B_1(f_i))=\{f_i, a\}$ &  \\\cline{2-4}
				&  \multicolumn{3}{l}{$T(V(G))=T(B_1(a))=\{t_i, t_{j}\}$, for $i, j\in [n]$, $i\neq j$ such that there is no} \\
				&  \multicolumn{3}{l}{ $k\in [m]$ where both $x_{i}$ and $x_j$ appear in $c_k$.\footnote{We can assume the existence of such a pair $i, j$, because introducing an artificial variable and a unique clause in which it occurs gives an equivalent instance. Indeed, the artificial variable appears only once and its assignment can be chosen so that the added clause is satisfied without the rest of the formula being affected.}}  
			\end{tabular}
			\caption{Positive teaching map for $\mathcal{B}$.}\label{tab:map}
		\end{table}	
	
	To this end, we will refer to Table~\ref{tab:check} which intuitively describes which vertices are used to distinguish each pair of balls of radius $1$. For any pair $u, z\in V(G)$, the corresponding entry in the table contains either $v$ or a given vertex $w\in V(G)$. If it contains a vertex $w\in V(G)$, observe that $w\in T(B_1(u))\cup T(B_1(z))$ by Table~\ref{tab:map}, and $w\notin B_1(u)\cap B_1(z)$. In the remaining cases, \emph{i.e.}, \emph{(1)}~$s^*_k, r_i^*$; \emph{(2)}~$s^{*}_k, s_k^{**}$; \emph{(3)}~$s^{*}_k, s_k'$; \emph{(4)}~$s^{*}_k, s_k''$, we need to make it explicit which vertex $v$ stands for. 
	For pairs \emph{(2)}~$s^{*}_k, s_k^{**}$; \emph{(3)}~$s^{*}_k, s_k'$; \emph{(4)}~$s^{*}_k, s_k''$, $v$ can be any vertex in $N(s_k^*)\cap A$ (by construction, there are $3$ such vertices) since the neighborhoods of each of $s_k^{**}$, $s_k'$, and $s_k''$ do not intersect $A$.
	
	For the last type of pair \emph{(1)}, we will use the fact that $\tau$ is a satisfying assignment.
	For $k\in[m]$ and $i\in [n]$, assume w.l.o.g. that $c_k$ is satisfied by the assignment $\tau(x_i)=\true$.
	Then, $T(B_1(s_k^*))=\{s_k', f_i\}$ and $T(B_1(r_i^*))=\{r_i', t_i\}$.
	With such an assignment of teaching sets, for any pair of type $s^*_k, r_i^*$ we have that $t_i\in T(B_1(s^*_k))\cup T(B_1(r^*_i))=\{s_k', f_i\}\cup \{r_i', t_i\}$, and $t_i \notin B_1(s^*_k) \cap B_1(r^*_i)$.
	
	\begin{table}[ht]\renewcommand{\arraystretch}{1.2}
		\subcaptionbox{Here $i, j\in [n]$, $k\in [m]$, and $v\in \{t_i, f_i\}\setminus N(s^*_k)$. According to Table~\ref{tab:map}, $\{t_i, f_i\}\setminus N(s^*_k)=t_i$ if $x_i=\true$ satisfies $c_k$; and $\{t_i, f_i\}\setminus N(s^*_k)=f_i$ if $x_i=\false$ satisfies $c_k$.
			\label{tab:main1}}{
			\begin{tabular}[t]{c|ccccc|c|cc}
				& $s_k^*$ & $s_k^{**}$ & $s_k^{***}$ & $s_k'$ & $s_k''$ & $s_0'$ & $t_i$ & $f_i$ \\\hline
				$r_i^*$		& $v$ & $r_0'$ & $r_0'$ & $r_i'$ & $r_i'$ & $r_i'$ & $r_i'$ & $r_i'$ \\ 
				$r_i^{**}$	& $s_0'$ & $r_0'$ & $r_0'$& $s_0'$ & $s_0'$ & $r_i'$ & $r_i'$ & $r_i'$ \\
				$r_i^{***}$	& $s_0'$ & $r_0'$ & $r_0'$ & $s_0'$ & $s_0'$ & $r_i''$ & $r_i''$ & $r_i''$ \\
				$r_i'$		& $s_k'$ & $r_0'$ & $r_0'$ & $r_i'$ & $r_i'$ & $r_i'$& $r_i'$ & $r_i'$ \\
				$r_i''$		& $s_k'$ & $r_0'$ & $r_0'$ & $r_i''$ & $r_i''$ & $r_i''$ & $r_i''$ & $r_i''$ \\\hline
				$r_0'$		& $s_k'$ & $s_k'$ & $s_k''$ & $s_k'$ & $s_k''$ & $r_0'$ & $r_0'$ & $r_0'$ \\\hline
				$t_j$		& $s_k'$ & $s_k'$ & $s_k''$ & $s_k'$ & $s_k''$ & $s_0'$ & $t_i$ & $f_i$ \\
				$f_j$		& $s_k'$ & $s_k'$ & $s_k''$ & $s_k'$ & $s_k''$ & $s_0'$ & $t_i$ & $f_i$ \\ 
		\end{tabular}}
		\hfill
		\subcaptionbox{Here, $k, l\in [m]$; filled cells correspond to the case $k = l$, and the others to $k\neq l$. Here, $v$ is any vertex in $N(s^*_l)\cap A$. The table for vertices of the variable force-gadgets is defined similarly, interchanging all $s$ and $r$ symbols.\label{tab:main2}}{
			\begin{tabular}[t]{c|ccccc}
				& $s_k^*$ & $s_k^{**}$ & $s_k^{***}$ & $s_k'$ & $s_k''$ \\\hline
				$s_l^*$		& $s_k'$ & \cg $v$ & \cg $s_k'$ & \cg $v$ & \cg $v$ \\ 
				$s_l^{**}$	& $s_k'$ & $s_k'$ & \cg $s_k'$& \cg $r_0'$ & \cg $r_0'$ \\
				$s_l^{***}$	& $s_k'$ & $s_k'$ & $s_k''$ & \cg $r_0'$ & \cg $r_0'$ \\
				$s_l'$		& $s_k'$ & $s_k'$ & $s_k''$ & $s_k'$ & \cg $s_k'$ \\
				$s_l''$		& $s_k'$ & $s_k'$ & $s_k''$ & $s_k'$ & $s_k''$ \\\hline
				$s_0'$		& $s_0'$ & $s_0'$ & $s_0'$ & $s_0'$ & $s_0'$
		\end{tabular}}\caption{For each $u, z\in V(G)$, in a cell at the intersection of the corresponding row and column, we place a vertex $w\in V(G)$ such that $w\in T(B_1(u))\cup T(B_1(z))$ and $w\notin B_1(u)\cup B_1(z)$ (according to the teaching map defined by Table~\ref{tab:map}).}\label{tab:check}
	\end{table}
	
	Finally, we check that the teaching set for $T(V(G))$ distinguishes $V(G)$ from all the other balls in $G$. According to Table~\ref{tab:map}, for $i, j\in [n]$ and $i\neq j$, $T(V(G))=\{t_i, t_j\}$ and there is no $k\in [m]$ such that both $x_i$ and $x_j$ appear in $c_k$. The last condition guarantees us that, for any $u\in V(G)\setminus \{a\}$, there is a vertex (either $t_i$ or $t_j$) that is in $T(V(G))$ but not in $B_1(u)$ as there is no clause force-gadget that would be attached to both $t_i$ and~$t_j$.
	
	Thus, we showed that for any pair of balls in $\mathcal{B}$, both necessary conditions for the defined teaching sets hold. Hence, the defined positive teaching map is non-clashing for $\mathcal{B}$.
\end{proof}

\begin{lemma}\label{lemma:correcntes2}
	If $(G,2)$ is a YES-instance of \strictp, then $\phi$ is a YES-instance of 3-\SAT.
\end{lemma}

\begin{proof}
	Let $T$ be a positive non-clashing teaching map of dimension $2$ for the set $\mathcal{B}$ of all balls in~$G$.
	By definition, for each pair of distinct balls $B_1, B_2\in\mathcal{B}$, there is $w\in T(B_1) \cup T(B_2)$ such that $w\notin  B_1 \cap B_2$. 
	For each $i\in [n]$, consider a pair of vertices $r^*_i$ and $r^{***}_i$.
	By the construction, $B_1(r^{***}_i)\subset B_1(r^{*}_i)$ and $r_i'$ is the unique vertex in $B_1(r^{*}_i)$ that is not in $B_1(r^{***}_i)\cap B_1(r^{*}_i)$. Thus, since $T$ is a positive non-clashing teaching map for $\mathcal{B}$, $r_i'$ is in $T(B_1(r_i^*))$.
	Now, consider a pair of vertices $r^*_i$ and $r^{**}_i$. Similarly, $B_1(r^{**}_i)\subset B_1(r^{*}_i)$ and $\{t_i, f_i\}=B_1(r^{*}_i)\setminus B_1(r^{**}_i)$. So, for $T$ to distinguish $B_1(r^{**}_i)$ and $B_1(r^{*}_i)$, exactly one of $t_i$ and $f_i$ (as $T(B_1(r_i^*))$ already contains $r_i'$) is in $T(B_1(r_i^*))$.
	The same arguments work for clause force-gadgets, by symmetry of the construction, and we obtain that, for each $k\in [m]$, there is $v\in B_1(s_k^*)\cap A$ such that $T(B_1(s_k^*))=\{s_k', v\}$.
	
	Now, let us use the fact that for $i\in [n]$ and $k\in [m]$, the balls $B_1(r_i^*)$ and $B_1(s_k^*)$ are distinguished by $T$.
	If $B_1(s_k^*)\cap \{t_i, f_i\}=\emptyset$, whichever of $t_i$, $f_i$ that is in $B_1(r_i^*)$ distinguishes the two balls.
	However, as we have shown before, for each $k\in [m]$, $|T(B_1(s_k^*))\cap A|=1$. So, there exists $i\in [n]$ such that $B_1(s_k^*)\cap \{t_i, f_i\}\neq \emptyset$. W.l.o.g., let us assume that $B_1(s_k^*)\cap \{t_i, f_i\}=f_i$ (which means that $T(B_1(s_k^*))=\{s_k', f_i\}$).
	As a result, the only valid option for $B_1(r_i^*)$ to be distinguished from $B_1(s_k^*)$ is that $T(B_1(r_i^*))=\{r_i', t_i\}$. In the other symmetric case where $B_1(s_k^*)\cap \{t_i, f_i\}=t_i$, we obtain $T(B_1(r_i^*))=\{r_i', f_i\}$.
	
	Let us now define an assignment $\tau: \mathcal{X}\rightarrow \{\true, \false\}$ in the following way. For each $i\in [n]$, if $t_i\in T(B_1(r_i^*))$, we set $\tau(x_i)=\true$;
	otherwise $f_i\in T(B_1(r_i^*))$ and we set $\tau(x_i)=\false$. 
	Let us show that $\tau$ indeed satisfies the 3-\SAT\ instance $\phi$.
	As we discussed above, for each $k\in [m]$, $T(B_1(s_k^*))$ has an intersection with $A$ in exactly one vertex, w.l.o.g., let it again be that $T(B_1(s_k^*))\cap A = \{f_i\}$. Then, $T(B_1(r_i^*))\cap A = \{t_i\}$. So, $\tau(x_i)=\true$.
	By our reduction, $s_k^*$ is adjacent to $f_i$ if assigning $\false$ to $x_i$ \textbf{does not} satisfy $c_k$, while assigning $\true$ to $x_i$ would.
	Thus, the assignment $\tau$ satisfies all the $m$ clauses of the initial 3-\SAT\ instance $\phi$.	
\end{proof}

The proof of Theorem~\ref{thm:split2} then follows from Lemma~\ref{lemma:correcntes1} and Lemma~\ref{lemma:correcntes2}. In particular, they prove that there is a polynomial-time reduction which transforms any instance of 3-\SAT\ with $n$ variables and $m$ clauses into an equivalent instance $(G,2)$ of \strictp\ where $|V(G)|=\bigoh(n+m)$ and $G$ is a split graph of diameter~$2$. The properties of this reduction also allow us to infer more precise algorithmic lower bounds. In particular, since an algorithm solving \strictp\ in $2^{o(|V(G)|\cdot d\cdot k)}$ time would allow us to solve 3-\SAT\ in $2^{o(n+m)}$ time:

\begin{theorem}
\label{thm:ETHlower}
Unless the Exponential Time Hypothesis fails, there is no algorithm solving \strictp\ in time $2^{o(|V(G)|\cdot d\cdot k)}$, where $d$ and $k$ are the diameter of $G$ and the target positive non-clashing teaching dimension of the instance, respectively.
\end{theorem}

We complement this lower bound with a refined upper bound for the more general \normalp:

\begin{proposition}\label{pro:exactalgo}
\normalp\ can be solved in $2^{\bigoh(|V(G)|\cdot d\cdot k\cdot \log |V(G)|)}$ time.
\end{proposition}

\begin{proof}
We can assume that $G$ is connected, as otherwise we can solve \normalp\ independently on each of the connected components of $G$. For any $v\in V(G)$ and $r\in \mathbb{N}$, there are at most $\binom{|V(G)|}{k}=\bigoh(2^{k\cdot \log |V(G)|})$ possible choices for $T(B_r(v))$, and there are at most $\bigoh(|V(G)|\cdot d)$ unique balls in $G$. Due to the latter, for each possible teaching map, we can check in polynomial time whether it is a positive teaching map that satisfies the non-clashing teaching property. Thus, there is a brute-force algorithm running in $2^{\bigoh(|V(G)|\cdot d\cdot k\cdot \log |V(G)|)}$ time. 
\end{proof}

\section{Fixed-Parameter Tractability via Vertex Integrity}

Given that \normalp\ is \NP-hard, it is natural to ask whether the problem can be solved efficiently on inputs exhibiting some well-defined structural properties. In this section, we establish the fixed-parameter tractability of \normalp\ when parameterized by the vertex integrity of the input graph.
Consider an instance $(G,\mathcal{B},k)$ of \normalp\ and let $p$ be the vertex integrity of $G$. As the first step, we invoke the known algorithm to compute a ``witness'' for the vertex integrity in time $p^{\bigoh(p)}|V(G)|$~\cite{fellows1989immersion}, \emph{i.e.}, a set $X\subset V(G)$ such that $|V(H)\cup X|\leq p$ for each connected component $H$ of $G - X$. Let $\mathcal{H}$ denote the set of connected components of $G-X$.
To make use of the vertex integrity of $G$, we will partition the components of $\mathcal{H}$ into a parameter-bounded number of equivalence classes such that the elements belonging to the same class share some structural properties that will allow us to consider them, to some extent, interchangeable.

\begin{definition}\label{def:twinblocks}
Two subgraphs $H,H'\in \mathcal{H}$ are \emph{twin-blocks} with respect to $\mathcal{B}$, denoted $H\simB H'$, if there exists an isomorphism $\alpha_{H,H'}$ from $H$ to $H'$ with the following properties:
\begin{itemize}
\item 
for each $u\in V(H)$ and $v\in X$, $uv\in E(G)$ if and only if $\alpha_{H,H'}(u)v\in E(G)$, and
\item for each $u \in V(H)$ and $r\in \mathbb{N}$, $B_r(u)\in \mathcal{B}$ if and only if $B_r(\alpha_{H,H'}(u))\in \mathcal{B}$.
\end{itemize}
\end{definition}

Intuitively, $H\simB H'$ if and only if there is a bijection $\alpha_{H,H'}$ between the vertices of the two subgraphs which preserves (1) edges inside $H$ and $H'$, (2) edges to $X$, and (3) the existence of balls in $\mathcal{B}$ centered at the vertices of $H$ and $H'$. 
We refer to $\alpha_{H,H'}$ as the \emph{canonical isomorphism} between the two twin-blocks at hand, and if multiple choices of $\alpha$ exist, we choose and fix one arbitrarily; we drop the indices of $\alpha$ when the subgraphs are clear from the context. 
Observe that for any choice of $H$ and $H'$, $H\simB H'$ can be tested in time at most $p^{\bigoh(p)}$ by enumerating all possible choices of $\alpha$. 

Clearly, $\simB$ is an equivalence relation and we denote by $[H]\eqB$ the equivalence class containing $H$. For $u\in V(H)$, we further define $[u]\eqB=\{\alpha_{H,H'}(u)~|~H'\in [H]\eqB\}$, and similarly for $B_r(u)\in \mathcal{B}$, $[B_r(u)]\eqB=\{B_r(u')\in \mathcal{B}~|~u'\in [u]\eqB\}$; intuitively, these refer to the sets of counterparts of $u$ and $B_r(u)$ in the equivalence class, respectively. For brevity, we overload the notation $\simB$ and use $v\simB w$ (or $B_r(v) \simB B_r(w)$) to express that $v\in [w]\eqB$ (or $B_r(v) \in [B_r(w)]\eqB$, respectively).

\begin{observation}\label{obs:numbclass}
The number of equivalence classes on $\mathcal{H}$ defined by $\simB$ is at most $2^{\bigo(p^3)}$.
\end{observation}

\begin{proof}
All graphs in $\mathcal{H}$ have size at most $p$, which means that the number of non-isomorphic graphs in $\mathcal{H}$ can be trivially upper-bounded by $p\cdot 2^{p^2}$. Since $|X|<p$, there are also at most $p^2$ possible edges between $X$ and any $H\in \mathcal{H}$ in $G$. Lastly, the number of balls centered in $H\in \mathcal{H}$ is bounded above by the number of vertices in $H$ times the diameter of $G$, which is $p\cdot \bigo(p^2) = \bigo(p^3)$. Indeed, the diameter of a connected graph with vertex integrity $p$ is at most $\bigo(p^2)$, since the parameter does not increase by taking induced subgraphs and the vertex integrity of a path of length $j$ is $\bigo(\sqrt{j})$.
Combining these elements, we can upper-bound the total number of equivalence classes by $p\cdot 2^{p^2} \cdot 2^{p^2} \cdot 2^{\bigo(p^3)} = 2^{\bigo(p^3)}$.
\end{proof}

While the equivalence relation $\simB$ is defined based on the input (in particular, $G$ and $\mathcal{B}$), our proof requires also considering a more refined equivalence relation based on the structure of a hypothetical positive non-clashing teaching map. Toward this, we use the following notion to capture how a hypothetical teaching set interacts with the balls centered in the components of $\mathcal{H}$.

\begin{definition}\label{def:blueprints}
The \emph{blueprint} $S$ of a teaching set $T(B)$ for a ball $B=B_r(u)$ centered in $H\in \mathcal{H}$ is a tuple $(S_X, S_H, S_f)$ composed of:\begin{enumerate}
\item the set $S_X = T(B)\cap X$,
\item the set $S_H = T(B)\cap V(H)$,
\item the set $S_f=\{f_{H_0}~|~H_0\in \mathcal{H}\}$ of functions, where for each $H_0$ the function $f_{H_0}:V(H_0)\rightarrow \{0,1,2\}$ specifies whether for a vertex $v\in V(H_0)$, the set $([v]\eqB\cap T(B))\setminus V(H)$ of counterparts of $v$ \emph{outside of $H$} has size $0$, $1$ or at least $2$.
\end{enumerate}
\end{definition}

Intuitively, the blueprint specifies how the teaching set for $B$ interacts with (1) the set $X$ and (2) the vertices inside $H$ itself; for the rest of the graph, the blueprint also counts how many ``equivalent'' vertices it contains from each equivalence class of $\mathcal{H}$, \emph{but only up to $2$}.
At this point, it may not be clear why we do not differentiate between any size greater than $2$; the reason is that if the actual size is $3$ or more, there are superfluous elements in the teaching set, as we prove below. In fact, we prove a more general statement which holds regardless of whether vertices in $H$ are counted or not.

\begin{lemma}\label{lem:max2equiv}
Let $u\in V(G)$, $B=B_r(u)\in \mathcal{B}$, and $T$ be a positive non-clashing teaching map for~$\mathcal{B}$. Suppose there exist $H_0\in \mathcal{H}$ and $v\in V(H_0)$ such that $|[v]\eqB\cap T(B)|\geq 3$. Then, there exists $z$ in $[v]\eqB\cap T(B)$ such that removing $z$ from $T(B)$ yields a positive non-clashing teaching map for $\mathcal{B}$.
\end{lemma}

\begin{proof}
Let $v$ and $H_0$ satisfy the premise, and let $z_1,z_2,z_3$ be three distinct elements of $\{w\in T(B) ~|~ w\in V(H'), H'\simB H_0, \alpha(w)=v\}$. We denote by $H_1$ ($H_2$, $H_3$, resp.) the component of~$\mathcal{H}$ containing $z_1$ ($z_2$, $z_3$, resp.). By definition, these components are disjoint since the vertices are $\simB$-equivalent, and thus, $u$ can be in at most one of $H_1,H_2,H_3$ (and possibly none, \emph{i.e.}, $u$ could be in some other component or in $X$). Without loss of generality, we assume that $u$ is not in $V(H_3)$, and claim that removing $z_3$ from $T(B)$ results in a positive non-clashing teaching map $T'$ for $\mathcal{B}$.

We prove this claim as follows. Toward a contradiction, suppose there is a ball $B'=B_{r'}(u')$ such that $T'$ does not satisfy the non-clashing condition for $B$ and $B'$. Then, $z_3$ was the only vertex in $T(B)\cup T(B')$ that was not contained in $B\cap B'$, and hence, $z_3\notin B'$ and $z_1, z_2\in B'$. Therefore, $d(u', z_1) < d(u', z_3)$, which implies that $u'\in V(H_1)$ since $z_1\simB z_3$. However, we also have that $d(u', z_2) < d(u', z_3)$, and so, $u'\in V(H_2)$, which is a contradiction since $V(H_1)\cap V(H_2)=\emptyset$. Thus, such a ball $B'$ does not exist, and we can safely remove $z_3$ from $T$, proving the lemma.
\end{proof}

As a corollary of Lemma~\ref{lem:max2equiv}, we can obtain an upper bound on the positive non-clashing teaching dimension of our instance, which will be useful later in this section.

\begin{corollary}\label{cor:boundedsol}
Let $G$ be a graph with vertex integrity $p$ and $\mathcal{B}$ an arbitrary set of balls of $G$. 
Then, the positive non-clashing teaching dimension of $\mathcal{B}$ is upper-bounded by a function $s(p)= 2^{\bigo(p^3)}$.
\end{corollary}

\begin{proof}
For any ball $B\in \mathcal{B}$ and positive non-clashing teaching map $T$ for $\mathcal{B}$, if Lemma~\ref{lem:max2equiv} can be applied, then there is an unnecessary vertex in $T(B)$ that can be removed. Thus, without loss of generality, consider a positive non-clashing teaching map $T$ for $\mathcal{B}$ for which Lemma~\ref{lem:max2equiv} cannot be applied for any $B\in \mathcal{B}$. For each $B\in \mathcal{B}$ and each of the at most $2^{\bigo(p^3)}$ equivalence classes (by Observation~\ref{obs:numbclass}) of~$\mathcal{H}$ for $\simB$, $T(B)$ can contain at most $2p$ vertices from that equivalence class, as otherwise it would contain $3$ vertices that are equivalent according to $\simB$. Thus, the total size of $T(B)$ is upper-bounded by $|X| + 2^{\bigo(p^3)}\cdot 2p = 2^{\bigo(p^3)}$.
\end{proof}

Returning to the notion of blueprints defined earlier, we can now formalize a refined notion of equivalence that also takes into account the behavior of a hypothetical solution.

\begin{definition}
Given $B=B_r(u)$ centered in $H$, and $H'\simB H$, we say that $B$ and $B'=B_r(\alpha(u))$ \emph{share the same blueprint} if, for $(S_X, S_H, (f_{H_0})_{H_0\in \mathcal{H}})$ and $(S'_X, S'_H, (f'_{H_0})_{H_0\in \mathcal{H}})$, the respective blueprints of $B$ and $B'$, the following hold:
\begin{itemize}
\item $S_X = S'_X$,
\item $\forall v\in V(H), v\in S_H \Leftrightarrow \alpha(v) \in S'_H$, and
\item $\forall H_0\in \mathcal{H}, f_{H_0}=f'_{H_0}$.
\end{itemize}
\end{definition}

\begin{definition}\label{def:perfect_eq}
Given a positive non-clashing teaching map $T$ for $\mathcal{B}$, we say that two components $H,H'\in \mathcal{H}$ are \emph{perfectly-equivalent twin-blocks} (or simply \emph{perfectly equivalent}) with respect to $T$ on $\mathcal{B}$, denoted $H\simBS H'$, if $H\simB H'$ and, for all $u\in V(H)$, $r\in \mathbb{N}$, and $B_r(u)\in \mathcal{B}$, it holds that $B_r(u)$ and $B_r(\alpha(u))$ share the same blueprint. Note that $\simBS$ is an equivalence relation.
\end{definition}

Crucially, we can also bound the number of equivalence classes for this refined equivalence defined w.r.t.\ a hypothetical teaching map. We denote the equivalence class of $H$ by $[H]\equivBS$.

\begin{observation}\label{obs:numbclassperfect}
The number of equivalence classes on $\mathcal{H}$ defined by $\simBS$ is at most $2^{2^{\bigo(p^3)}}$.
\end{observation}

\begin{proof}
By Observation~\ref{obs:numbclass}, there are $2^{\bigo(p^3)}$ possibilities for the equivalence class on $\mathcal{H}$ defined by~$\simB$.
For any ball, there are then at most $2^p$ possibilities for the set $S_X$, as well as at most $2^p$ for the set $S_H$, and for a given $H_0$, at most $3^p$ for $f_{H_0}$. However, these choices for $(f_{H_0})_{H_0 \in \mathcal{H}}$ are not independent: $H'\simB H'' \Rightarrow f_{H'}=f_{H''}$. Thus, there are in total $2^p\cdot 2^p\cdot (3^p)^{2^{\bigo(p^3)}}=2^{2^{\bigo(p^3)}}$ possible blueprints for a given ball.
There are at most $\bigo(p^3)$ balls centered in a given $H$ (for details, see proof of Observation~\ref{obs:numbclass}).
Putting everything together gives us an upper bound for the number of equivalence classes for $\simBS$ of $2^{\bigo(p^3)} \cdot (2^{2^{\bigo(p^3)}})^{\bigo(p^3)}=2^{2^{\bigo(p^3)}}$.
\end{proof}

With this refined equivalence in hand, we proceed to the second milestone required for our algorithm: a structural result establishing the existence of ``well-behaved teaching maps''. Roughly speaking, by ``well-behaved'' we mean teaching maps which only use examples (\emph{i.e.}, vertices) from $X$ and a bounded number of ``core'' components in $\mathcal{H}$, with some controlled exceptions. Toward formalizing this, let us fix an arbitrary total ordering $\prec$ on the components in $\mathcal{H}$. 

\begin{definition}\label{def:core}
For $B\in \mathcal{B}$ and $x\in T(B)\cap V(H)$ for some $H\in \mathcal{H}$, the $(B,x)$\emph{-core} $\coreBx$ contains the first $s(p)+1$ (see Corollary~\ref{cor:boundedsol}) elements w.r.t.\ $\prec$ 
of $\{H'\in [H]\eqB~|~\exists B'\in [B]\equivBS, \alpha(x)\in T(B')\cap V(H')\}$.
The $B$\emph{-core} is the union of the $(B,x)$-cores for all such $x$, and the \emph{core} $\core$ is the union for all $B\in \mathcal{B}$ of the $B$-cores.
\end{definition}

\begin{definition}\label{def:ext_core}
A component $H\in \mathcal{H}$ is in the \emph{1-extended core} $\extcore$ if $H\in \core$ or there exists $B\in \mathcal{B}$ and $x\in V(H)$ such that (1) $B$ is centered in $\core$ or in $X$, and (2) $x\in T(B)$. Going one step further, a component $H\in \mathcal{H}$ is in the \emph{2-extended core} $\extcorebis$ if $H\in \extcore$ or there exists $B\in \mathcal{B}$ and $x\in V(H)$ such that (1) $B$ is centered in $\extcore$ or in $X$, and (2) $x\in T(B)$. 
\end{definition}

For $A\in \{\core, \extcore, \extcorebis\}$, note that $A$ is a subset of $\mathcal{H}$, but by slight abuse of notation, we will write $x\in V(A)$ for $x\in V(G)$ to denote $\exists H\in A$ with $x\in V(H)$ when not leading to any ambiguity. It is not difficult to see that the definitions yield the following.

\begin{observation}\label{obs:ext_core}
The number of components in $\extcorebis$ is upper-bounded by $c(p)= 2^{2^{\bigo(p^3)}}$.
\end{observation}

\begin{proof}
Any $(B,x)$-core $\coreBx$ contains at most $s(p)+1$ components by definition. Thus, any $B$-core contains at most $\bigo(s(p)^2)$ components (see Corollary~\ref{cor:boundedsol}). Since two $\simB$-equivalent balls from perfectly-equivalent twin-blocks have the same $B$-core, we get that $\core$ has at most $\bigo(p^3)\cdot 2^{2^{\bigo(p^3)}}\cdot \bigo(s(p)^2)=2^{2^{\bigo(p^3)}}$ components (cf. Observation~\ref{obs:numbclassperfect}).
Adding $p$ to this value before multiplying it by $\bigo(p^3)$ gives a bound on the number of balls centered in $\core$ or $X$, and with an additional factor $s(p)$ we get the maximum size of the union of their teaching sets, which is an upper bound for the number of components in $\extcore$. We can reiterate this to obtain an upper bound for $|\extcorebis|$ which is:
$s(p)\cdot \bigoh(p^3)\cdot \left(p+ s(p)\cdot \bigoh(p^3) \cdot (p+ 2^{2^{\bigo(p^3)}})\right)=2^{2^{\bigo(p^3)}}$.
\end{proof}

We can now formalize the notion of ``well-behaved'' teaching maps via the property of compactness.

\begin{definition}\label{def:compact}
A teaching map $T$ is \emph{compact} if, for any $H\in \mathcal{H}\setminus \extcorebis$ and any $x\in V(H)$, each ball whose teaching set contains $x$ is centered in $H$.
\end{definition}

Our next aim is to prove that we can safely restrict our attention to computing compact teaching maps only; this is stated in Lemma~\ref{lem:exists_compact} later on. Before we can establish that result, we first prove two auxiliary statements.

\begin{lemma}\label{lem:tradeoff_swap}
Let $T$ be a positive non-clashing teaching map for $\mathcal{B}$. Let $B\in \mathcal{B}$ be centered in $H_B\in \mathcal{H}\setminus \extcore$ with $x\in T(B)\cap V(H)$ for some $H\in \mathcal{H}\setminus H_B$ and $H_i\in \coreBx$. By definition, there exists $x_i\in[x]\eqB \cap V(H_i)$. If replacing $x$ by $x_i$ in $T(B)$ creates a conflict, then there exists $z_i\in T(B)\cap V(H_i)$.
\end{lemma}

\begin{proof}
Let $B'$ be a ball creating a conflict with $B$ when we replace $x$ by $x_i$ in $T(B)$. For a conflict to happen, it means that $x_i\in B'$ whereas $x\notin B'$. Since $x_i\in [x]\eqB$, then $B'$ has to be centered in $H_i$. Let us consider $B'_0 \in [B']\eqB$ centered in $H$. Prior to the replacement, $B'_0$ has no conflict with~$B$, and thus, there are two possibilities:

\emph(1) Let us first consider the case where there exists $z\in T(B'_0)$ that distinguishes $B'_0$ from $B$. By definition, $z\notin B$. Since $B'\simB B'_0$ and $H'\simBS H$, there exists $z'\in T(B')$, such that $z'\in [z]\eqB$. As $z'$ does not distinguish $B'$ from $B$, it holds that $z'\in B$. However, since $z'\simB z$, this means that $z'\in H_B$, which leads to a contradiction since $H_B \notin \extcore$.

\emph(2) Otherwise, there exists $z\in T(B)$ that distinguishes $B$ from $B'_0$, which means that $z\notin B'_0$. However, $z\in B'$ since there is a conflict with $B'$. As $B'\simB B'_0$ and $H'\simBS H$, then the facts that $z\notin B'_0$ and $z\in B'$ imply that $z\in H_i$. Hence, we found the $z_i$ whose existence we claimed.
\end{proof}

\begin{lemma}\label{lem:exists_swap}
Let $T$ be an optimal positive non-clashing teaching map for $\mathcal{B}$, with $B\in \mathcal{B}$ centered~in $H_B\in \mathcal{H}\setminus \extcore$. For $x\in T(B)\cap V(H)$ for some $H\in \mathcal{H}\setminus H_B$, if $|\coreBx|=s(p)+1$, then there are $H_i\in \coreBx$ and $x_i\in[x]\eqB \cap V(H_i)$ such that replacing $x$ by $x_i$ in $T(B)$ creates no conflict.
\end{lemma}

\begin{proof}
We prove the lemma by contradiction. Suppose that for all $x_i\in [x]\eqB$ of $\coreBx$, a conflict is created. By Lemma~\ref{lem:tradeoff_swap}, for all $H_i$ in $\coreBx$, there exists $z_i\in T(B)\cap V(H_i)$. 
However, these $H_i$'s are all disjoint, and hence, $|T(B)|\geq |\coreBx|=s(p)+1$, which contradicts the maximum size of a teaching set established in Corollary~\ref{cor:boundedsol}.
This proves that there exist such an $H_i$ and $x_i$.
\end{proof}

\begin{lemma}\label{lem:exists_compact}
If $\mathcal{B}$ has positive non-clashing teaching dimension $k$, then it also admits a compact positive non-clashing teaching map of dimension $k$. 
\end{lemma}

\begin{proof}
This directly follows from the repeated application of Lemma~\ref{lem:exists_swap} to any positive non-clashing teaching map $T$ for $\mathcal{B}$ of dimension $k$.
Let us assume that we do not yet have a compact teaching map. By definition, there exists $H_v\in \mathcal{H}\setminus \extcorebis$, $v\in V(H_v)$, and $B_r(u)$ such that $v\in T(B_r(u))$ and $u\notin V(H_v)$. It cannot be that $B_r(u)$ is centered in $\extcore$, since otherwise $v\in T(B_r(u))$ would imply that $H_v\in \extcorebis$. We can then apply Lemma~\ref{lem:exists_swap} to $B_r(u)$ and $v$, and replace $v$ in $T(B_r(u))$ with some vertex in $\core$. This preserves the non-clashing quality of the teaching map, and does not increase the size of the teaching sets. Moreover, since the number of such pairs $v, B_r(u)$ strictly decreases with each application, we eventually reach the point where there is no such occurrence in the teaching map, meaning that the obtained positive non-clashing teaching map is compact.
\end{proof}

We now proceed to the crux of our algorithm: the proof that one can reduce the size of the instance $(G,\mathcal{B})$ without changing its positive non-clashing teaching dimension (formalized in Lemma~\ref{lem:towards_red}). Toward this, it will be useful to focus on how the balls in $\mathcal{B}$ interact with certain subgraphs of $G$.

\begin{definition}\label{def:ind_balls}
Let $G'$ be an induced subgraph of $G$. Then, $\mathcal{B}$ \emph{induces} the set $\mathcal{B'}$ of vertex sets w.r.t.~$G'$, where $\mathcal{B'}=\{B_r(u)\cap V(G')~|~  B_r(u)\in \mathcal{B} \land u\in V(G')\}$.
\end{definition}

We note that $\mathcal{B'}$ need not necessarily be a set of balls in $G'$ itself: for instance, it may well happen that for some $B\in \mathcal{B'}$ the vertex set $B\cap V(G')$ is not even connected. However, under some conditions on $G'$ that we will be able to guarantee, $\mathcal{B'}$ is, in fact, a set of balls in $G'$.

\begin{lemma}\label{lem:ind_balls}
Let $G$ be a graph, $\mathcal{B}$ a set of balls in $G$, and $G'$ an induced subgraph of $G$. If $V(G')$ contains $X$ and, for every $H\in \mathcal{H}$, either (1) $V(H)\subseteq V(G')$ or (2) $V(H)\cap V(G') = \emptyset \land \exists H', H''\in [H]\eqB, H'\neq H'', V(H')\cup V(H'')\subseteq V(G')$ holds, then the set $\mathcal{B'}$ induced by $\mathcal{B}$ w.r.t.\ $G'$ is a set of balls in $G'$. Moreover, there is a bijection between balls in $\mathcal{B'}$ and balls in $\mathcal{B}$ centered in $G'$.
\end{lemma}

\begin{proof}
The key element of this proof is that in such a graph $G'$, the distance $d_{G'}(u,v)$ between two vertices $u$ and $v$ is the same as the distance $d_G(u,v)$ between them in $G$. Indeed, for $u,v\in V(G')$, suppose that the shortest path between them in $G$ uses some vertices not present in $G'$. For a consecutive set of such missing vertices on the path, since they are missing in $G'$ and connected in $G$, they are not in $X$ and they are all in the same component $H\in  \mathcal{H}$. Since $H$ is missing in $G'$, the vertices directly before and after on the path are from $X$: these vertices exist because $u,v \in V(G')$. Since there is $H'\in [H]\eqB$ such that $V(H')\subseteq V(G')$, we will replace the vertices of the path in $H$ by the $\simB$-equivalent vertices of $H'$. The newly constructed path has the same length, is still a path by the definition of $H\simB H'$, and has strictly fewer vertices not in $G'$. Iterating this argument until we obtain a path with only vertices of $G'$, we prove that the distance between $u$ and $v$ is the same in $G'$ as it was in $G$. From now on, we denote this distance $d(u,v)$, without referring to the graph.

We are now ready to prove that the vertex sets in $\mathcal{B'}$ are indeed balls in $G'$. Let $B'\in \mathcal{B'}$. There exist $u\in V(G')$ and $B_r(u)\in \mathcal{B}$, such that $B'=B_r(u)\cap V(G')=\{v\in V(G'), d(u,v)\leq r\}$. Hence, $B'$ is indeed a ball in $G'$.
It remains to prove that this construction is injective. Let $u,v\in V(G')$ and $B_r(u), B_{r'}(v)\in \mathcal{B}$, such that $B_r(u)\cap V(G')=B_{r'}(v)\cap V(G')$ and assume there is $w\in V(G)\setminus V(G')$ such that $w\in B_r(u)$ and $w\notin B_{r'}(v)$. Let $H\in \mathcal{H}$ be such that $w\in V(H)$, and $H', H''\in [H]\eqB$ are the two corresponding components whose existence is required by (2), and $w'\in V(H')\cap [w]\eqB$ ($w''\in V(H'')\cap [w]\eqB$, resp.). If $w'\notin B_r(u)$, then it implies that $d(u,w)<d(u,w')$, and thus, $u\in V(H)$, which is a contradiction since $V(G')\cap V(H)=\emptyset$.
With the same argument for $w''$, we infer that $w', w''\in B_r(u)$, and hence, $w', w''\in B_{r'}(v)$ since $w',w''\in V(G')$. However, $w'\in B_{r'}(v)$ implies that $d(v,w')<d(v,w)$, which in turn implies that $v\in V(H')$. We also obtain $v\in V(H'')$ by the same reasoning with $w''$, and this leads to a contradiction since $V(H')\cap V(H'')=\emptyset$. 
We proved that it was impossible to find two such balls, which proves the claimed bijection between balls in $\mathcal{B'}$ and balls in $\mathcal{B}$ centered in $G'$.
\end{proof}

We define the \emph{reduced graph $G'$ of $G$} as the graph obtained from $G$ by removing all but $f(p)=c(p)+\binom{p\cdot c(p)+p}{s(p)}^{b(p)} +1$ twin-blocks from each large class of $\simB$, where $b(p)=\bigo(p^3)$ is the maximum number of distinct balls centered in any component of $\mathcal{H}$. 
Let $\mathcal{B'}$ be the set induced by $\mathcal{B}$ on $G'$ according to Definition~\ref{def:ind_balls}. By Lemma~\ref{lem:ind_balls}, $\mathcal{B'}$ is a set of balls in $G'$. The size of $G'$ is at most $g(p)=p+f(p)\cdot 2^{\bigo(p^3)}=2^{2^{\bigo (p^3)}}$. Thus, an optimal positive non-clashing teaching map for $\mathcal{B'}$ can be computed (\emph{e.g.}, by brute force) in time that depends only on $p$.
We now aim at proving that $\mathcal{B'}$ is equivalent to $\mathcal{B}$, \emph{i.e.}, that the optimal positive non-clashing teaching dimension is the same for both.

\begin{lemma}\label{lem:towards_red}
Suppose that there is a solution for $(G,\mathcal{B},k)$. Then, there exists a solution $T'$ for the reduced instance $(G',\mathcal{B'}, k)$.
\end{lemma}

\begin{proof}
Applying Lemma~\ref{lem:exists_compact} gives us the existence of a compact solution $T$. Up to renaming, $G'$ corresponds to $G$ in which we remove only components \textbf{not} in $\extcorebis$ for $T$. Let us define $T'$ as the restriction of $T$ to $\mathcal{B'}$---this is well-defined since each ball in $\mathcal{B'}$ corresponds to exactly one in $\mathcal{B}$ as per Lemma~\ref{lem:ind_balls}. 
It is easy to see that if two balls in $\mathcal{B'}$ were in conflict, then the corresponding balls in $\mathcal{B}$ would be as well. Indeed, the teaching sets are exactly the same, since we did not remove any component of $\extcorebis$ and $T$ is compact. Furthermore, the balls themselves can only be smaller in $\mathcal{B'}$, and thus, the element used by $T$ to distinguish between the two balls of $\mathcal{B}$ can be used by $T'$ to distinguish between the two balls in $\mathcal{B'}$. 
\end{proof}

The following lemma establishes that a solution for the reduced instance can be lifted to one for the original instance; note that this also ensures that our algorithm will be constructive.

\begin{lemma}\label{lem:from_red}
Let $k\in \mathbb{N}$ and $T'$ be a compact solution for the instance $(G', \mathcal{B'}, k)$. We can construct a compact solution $T$ for $(G,\mathcal{B}, k)$ in $2^{2^{\bigoh(p^3)}} \cdot |V(G)|^{\bigoh(1)}$ time.
\end{lemma}

\begin{proof}
We iteratively add back the missing components $(H_0,\dots, H_x)$ of $\mathcal{H}$ and update the set of balls accordingly. We start by choosing an arbitrary ordering $\prec$ on $\mathcal{H}$ such that: $\forall H,H'\in \mathcal{H}, V(H)\subseteq V(G') \land V(H')\cap V(G')=\emptyset \Rightarrow H\prec H'$. We then compute $\extcoreT$ for $T'$ and $\prec$ on $G'$, as it is used at each step. The time necessary for it depends only on $p$ and is $\bigoh(|\mathcal{B'}|\cdot s(p))$, as it suffices to check for each ball which components its teaching set uses, apart from the component where it is potentially centered. Indeed, as $T'$ is compact, $H$ being in $\extcoreT$ implies there is a ball not centered in $H$ using $x\in V(H)$ in its teaching set. We prove the lemma via the following claim.

\begin{claim}\label{cl:induction_add_copy}
 Let $0 \leq c < x$, $\mathcal{B^*}$ be the set of balls induced by $\mathcal{B}$ on $G^*=G[V(G')\bigcup_{0\leq i < c} V(H_i)]$, and let $T^*$ be a solution for $(G^*,\mathcal{B^*},k)$ such that $\extcoreT$ is the $2$-extended core for $T^*$. We can construct, in $2^{2^{\bigoh(p^3)}} \cdot |V(G)|^{\bigoh(1)}$ time, a solution $T^+$ for $(G^+,\mathcal{B^+},k)$, where $G^+=G[V(G^*)\cup V(H_{c})]$, $\mathcal{B^+}$ is induced by $\mathcal{B}$ on $G^+$, and $\extcoreT$ is the $2$-extended core for $T^+$. 
 \end{claim}

\begin{proof}
It is easy to observe that $\mathcal{B^*}$ is induced by $\mathcal{B^+}$ on $G^*$. Thus, by Lemma~\ref{lem:ind_balls}, there is a bijection between the balls in $\mathcal{B^*}$ and those of $\mathcal{B^+}$ which are centered in $G^*$. We can compute $\mathcal{B^+}$ in $|V(G)|^{\bigoh(1)}\cdot p^2$ time, since for the at most $|V(G)|\cdot p^2$ balls in $\mathcal{B^+}$, a breadth-first search on $G^+$ suffices to compute it. Note that for any ball $B^*\in \mathcal{B^*}$, the corresponding ball $B^+ \in \mathcal{B^+}$ is a superset of it, and we set $T^+(B^+)=T^*(B^*)$. Note that for any two balls whose teaching set we define in this way, there can be no conflict between them. Indeed, any element previously distinguishing their respective equivalent balls in $\mathcal{B'}$ is in $V(G')$, and thus, will still be at the same distance from each center, meaning it distinguishes the two balls in $\mathcal{B^+}$ as well.

However, not all balls in $\mathcal{B^+}$ are centered in $G^*$: we now need to define $T^+$ for the balls centered in $H_c$. 
Toward this end, we first need to identify some components in $G^*$ which will be useful to define the teaching map so that it is non-clashing.

We know that there are at least $f(p)$ components of $[H_{c}]\eqB$ in $G'$ by definition, and thus, in $G^*$. Since $f(p)=c(p)+\binom{p\cdot c(p)+p}{s(p)}^{b(p)} +1$, there are at least $\binom{p\cdot c(p)+p}{s(p)}^{b(p)} +1$ components outside of $\extcoreT$. Since all of these components are twin-blocks, they have the same balls (at most $b(p)$ of them), and by the pigeonhole principle, since each ball has a teaching set consisting of at most $s(p)$ vertices of the components of $\extcoreT \cup X$ and the component the ball is centered in (which contain at most $p\cdot c(p) + p$ vertices combined), at least 2 of them have, for each pair of $\simB$-equivalent balls, the exact same teaching sets in $\extcoreT \cup X$ as well as isomorphically identical teaching sets in their own components. Let us denote these two components as $H',H''$.

Let $u\in V(H_c)$, $r\in \mathbb{N}$ be such that $B^+_r(u)\in \mathcal{B^+}$. We ``copy'' the teaching set $T^*(B^*_r(v))$, where $v\in V(H') \cap [u]\eqB$. Formally, for $\alpha$ the canonical isomorphism from $H'$ to $H_c$, $T^+(B^+_r(u))= \{\alpha(w) \mid w\in T^*(B^*_r(v))\cap H' \} \cup \{w \mid w\in T^*(B^*_r(v)), w\notin H'\}$. Thus, from the point of view of any ball centered outside of $H'$ and $H_c$, $B^+_r(u)$ and $T^+(B^+_r(u))$ behave the same as $B^*_r(v)$ and $T^*(B^*_r(v))$, and so, there are no conflicts with such balls, as they were not in conflict with~$B^*_r(v)$. Two balls centered in $H_c$ are also distinguished from each other since this is the case for pairs of balls centered in $H'$. The last thing to check is that balls centered in $H_c$ are distinguished from balls centered in $H'$. Here, we use the fact that balls in $H''$ have the same teaching sets (up to isomorphism between vertices of $H'$ and $H''$) as those in $H'$, and so, if there was a conflict between balls centered in $H'$ and $H_c$, there would also be one between balls centered in $H'$ and~$H''$.

Thus, we managed to construct a teaching map $T^+$ for $\mathcal{B^+}$. Since the new teaching sets are imitations of pre-existing teaching sets in $T^*$, the size constraint is satisfied, and it is easy to check that the solution is still compact. Indeed, the teaching sets of balls centered in $H'$ contained only vertices of $H'$ and $\extcoreT$, and thus, those of balls centered in $H_c$ contain by construction only vertices of $H_c$ and $\extcoreT$. Furthermore, $H_c$ is not in $\extcorebis$ for $T^+$: indeed, $H_c\notin \core$ because there are already enough components $\simB$ equivalent to $H_c$ in $G'$, and all of them are preceding $H_c$ in the order $\prec$. Moreover, no vertex of $H_c$ has been added to the teaching sets of the balls not centered in $H_c$: by definition, $H_c\notin \extcorebis$. Thus, $\extcorebis=\extcoreT$.

Note that we can find $H'$ and $H''$ in time $f(p)^2\cdot (c(p)+p)^{s(p)\cdot \bigoh(p^3)}$, and then we construct teaching sets for the at most $b(p)$ balls in $H_c$, each of them being a simple copy of the equivalent teaching set in $H'$, which can be done in $\bigoh(s(p))$ time. Thus, the running time is $f(p)^2\cdot (c(p)+p)^{s(p)\cdot \bigoh(p^3)}\bigoh(s(p)) \cdot |V(G)|^{\bigoh(1)}= 2^{2^{\bigoh(p^3)}} \cdot |V(G)|^{\bigoh(1)}$ and we have proven the claim.
\end{proof}
We can now finish proving the lemma by induction using Claim~\ref{cl:induction_add_copy}. Since the hypothesis trivially holds for $c=0$ (by considering $G^*=G'$), we prove the claimed result, and the total running time follows from the fact that the number of components of $\mathcal{H}$ missing in $G'$ are at most $|V(G)|$.
\end{proof}

Finally, we obtain our main parameterized tractability result by combining the previous ingredients. In particular, in the proof, we construct $(G',\mathcal{B}',k)$, solve the problem there, and argue correctness.

\begin{theorem}\label{thm:fpt_vi}
\normalp\ is \FPT\ parameterized by the vertex integrity of the input graph $G$.
\end{theorem}

\begin{proof}
Let $G$ be a graph with vertex integrity $p$, $\mathcal{B}$ a set of balls in $G$, and $k\in \mathbb{N}$. We provide an algorithm computing a positive non-clashing teaching map for $\mathcal{B}$ of dimension at most $k$, or correctly outputing that none exists, in $q(p)\cdot |V(G)|^{\bigoh(1)}$ time, where $q$ is an elementary function.
It is useful to note that if $k\geq s(p)$, then Corollary~\ref{cor:boundedsol} implies there is always a solution. Moreover, replacing the value of $k$ by $s(p)$ for the rest of the algorithm does not hurt: the solution still exists and has dimension smaller than $k$. Hence, we assume in the rest of the proof that $k\leq s(p)$.

The first step of the actual algorithm is to compute the witness $X\subset V(G)$ for the vertex integrity and the corresponding set  $\mathcal{H}$ of connected components.
Next, we classify the elements of $\mathcal{H}$ w.r.t.~the equivalence classes defined by $\simB$ (see Definition~\ref{def:twinblocks}). Since there are at most $2^{\bigo(p^3)}$ equivalence classes and the equivalence between two components can be tested in $p^{\bigoh(p)}$ time, we can compute the equivalence classes in $|V(G)|\cdot2^{\bigo(p^3)}\cdot p^{\bigoh(p)}$ time with brute force.

We are now ready to compute the reduced graph $G'$ of $G$. We recall that to this end it suffices to remove some arbitrary components of $\mathcal{H}$ whose equivalence class are bigger than $f(p)$, which can be done in linear time. Recalling Definition~\ref{def:ind_balls}, the set $\mathcal{B'}$ induced by $\mathcal{B}$ on $G'$ can be computed in $|V(G')|\cdot \bigoh(p^2)\cdot |V(G)|\cdot |V(G')|$ time; indeed, it suffices to check for $u\in V(G')$ and $r\in \mathbb{N}$ (which is at most the diameter of the graph) whether $B_r(u)$ exists in $\mathcal{B}$, and then to compute the intersection of it with $V(G')$. 
The size of the instance $(G', \mathcal{B'}, k)$ is upper-bounded by a function of $p$, meaning that we can compute a positive non-clashing teaching map of dimension at most $k$ for $\mathcal{B'}$ in time depending only on $p$ (or determine that none exists)---see Proposition~\ref{pro:exactalgo}.

Using Lemma~\ref{lem:towards_red}, we know that if there is no such teaching map for $\mathcal{B'}$, there is also none of dimension $k$ for $\mathcal{B}$: the algorithm can safely output that no solution exists.
Conversely, if we obtain a positive non-clashing teaching map for $\mathcal{B'}$, then we can use Lemma~\ref{lem:from_red} to construct a positive non-clashing teaching map of dimension at most $k$ for $\mathcal{B}$ in $2^{2^{\bigoh(p^3)}}\cdot |V(G)|^{\bigoh(1)}$ time.
This concludes the proof, and the total running time can be upper-bounded by $2^{2^{\bigoh(p^3)}} \cdot |V(G)|^{\bigoh(1)}$.
\end{proof}

\section{Hardness for Classical Structural Parameterizations}\label{sec:whard}

In this section, we prove that \normalp\ is intractable w.r.t.~the feedback vertex number and pathwidth; the same then holds for generalizations of these measures like treewidth or clique-width.

\begin{theorem}\label{thm:hard-fvs}
	\normalp\ is \W\textup{[1]}-hard parameterized by $\fvs(G)+\pw(G)+k$. 
\end{theorem}

We establish Theorem~\ref{thm:hard-fvs} by describing a parameterized reduction below and proving its correctness in Lemma~\ref{lem:w-correctness}.
We reduce from a problem called {\sc NAE-Integer-3-SAT}, which is \W[1]-hard parameterized by the number of variables~\cite{BHML16}, and is defined as follows.

\defquestion{{\sc NAE-Integer-3-SAT}}{A set of clauses $\mathcal{C}$ over variables $\mathcal{X}$, and an integer $d$. Any clause $c\in\mathcal{C}$ has the form $x\leq c_x$, $y\leq c_y$, and $z\leq c_z$, where $c_x,c_y,c_z\in \{1,\ldots,d\}$.}{Is there a variable assignment $\mathcal{X}\rightarrow \{1, \ldots, d\}$ such that for each clause, either $1$ or $2$ of its three inequalities are satisfied?}

\noindent\textit{The Reduction.}
Given an instance $\phi$ of {\sc NAE-Integer-3-SAT}, we construct an instance $(G,\mathcal{B},k)$ of \normalp\ in polynomial time as follows, starting with the graph $G$.

	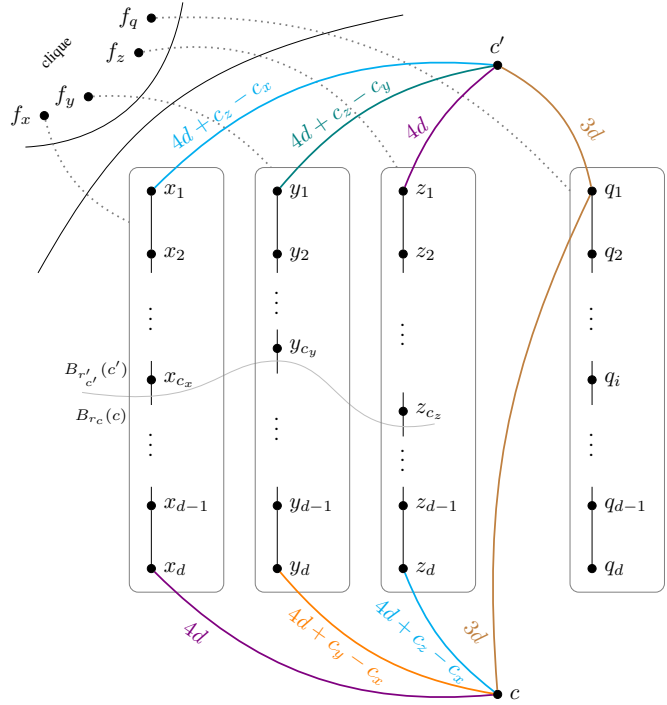
\begin{figure}[t]
	\scalebox{0.837}{\centering
		\begin{tikzpicture}
\def\y{3}\def\x{-.5}

\node[circ, label=right:{$x_1$}] (x1) at (\x, \y+6) {};
\node[circ, label=right:{$x_2$}] (x2) at (\x, \y+5) {};
\node[circ, label=right:{$x_{c_x}$}] (xc) at (\x, \y+3) {};
\node[circ, label=right:{$x_{d-1}$}] (xd1) at (\x, \y+1) {};
\node[circ, label=right:{$x_d$}] (xd) at (\x, \y) {};

\node[circ, label=right:{$y_1$}] (y1) at (\x+2, \y+6) {};
\node[circ, label=right:{$y_2$}] (y2) at (\x+2, \y+5) {};
\node[circ, label=right:{$y_{c_y}$}] (yc) at (\x+2, \y+3.5) {};
\node[circ, label=right:{$y_{d-1}$}] (yd1) at (\x+2, \y+1) {};;
\node[circ, label=right:{$y_d$}] (yd) at (\x+2, \y) {};

\node[circ, label=right:{$z_1$}] (z1) at (\x+4, \y+6) {};
\node[circ, label=right:{$z_2$}] (z2) at (\x+4, \y+5) {};
\node[circ, label=right:{$z_{c_z}$}] (zc) at (\x+4, \y+2.5) {};
\node[circ, label=right:{$z_{d-1}$}] (zd1) at (\x+4, \y+1) {};;
\node[circ, label=right:{$z_d$}] (zd) at (\x+4, \y) {};

\node[circ, label=right:{$q_1$}] (q1) at (\x+7, \y+6) {};
\node[circ, label=right:{$q_2$}] (q2) at (\x+7, \y+5) {};
\node[circ, label=right:{$q_i$}] (qc) at (\x+7, \y+3) {};
\node[circ, label=right:{$q_{d-1}$}] (qd1) at (\x+7, \y+1) {};;
\node[circ, label=right:{$q_d$}] (qd) at (\x+7, \y) {};

\foreach \v in {x,y,z,q} {
	\draw (\v1) -- (\v2) --++ (0,-.3);
	\node[draw=none,rotate=90] at ($.5*(\v 2) +.5*(\v c)$) {$\cdots$};
	\draw (\v d) -- (\v d1) --++ (0,.3);
	\node[draw=none,rotate=90] at ($.5*(\v d1) +.5*(\v c)$) {$\cdots$};
	\draw ($(\v c)-(0,.4)$) -- ($(\v c)+(0,.3)$);
}

\node[circ, label=above:{$c'$}] (c) at (\x+5.5, \y+8) {};
\node[circ, label=right:{$c$}] (c') at (\x+5.5, \y-2) {};

\draw [gray!50] ($(xc)+(-1.1, -.2)$) .. controls ($(xc)+(1, -.5)$) and ($(yc)+(-1, -.2)$) .. ($(yc)+(0, -.2)$);
\draw [gray!50] ($(yc)+(0, -.2)$) .. controls ($(yc)+(1, -.2)$) and ($(zc)+(-1, -.5)$) .. ($(zc)+(.5, -.2)$);

\node[draw=none] at ($(xc)+(-.8, -.575)$) {\scriptsize $B_{r_c}(c)$};
\node[draw=none] at ($(xc)+(-.875, .1)$) {\scriptsize $B_{r'_{c'}}(c')$};

\node[circ, label=left:{$f_y$}] (fx) at (\x-1, \y+7.5) {};
\node[circ, label=left:{$f_x$}] (fy) at ($(fx)+(-.7, -.3)$) {};
\node[circ, label=left:{$f_z$}] (fz) at ($(fx)+(.8, .7)$) {};
\node[circ, label=left:{$f_q$}] (fq) at ($(fx)+(1, 1.25)$) {};
\node (FX) at (-.1, 6) {
	\tikz {\draw[color = gray, rounded corners] (0, 0) rectangle ++(1.5,6.75);}
};
\node (FY) at (1.9, 6) {
	\tikz {\draw[color = gray, rounded corners] (0, 0) rectangle ++(1.5,6.75);}
};
\node (FZ) at (3.9, 6) {
	\tikz {\draw[color = gray, rounded corners] (0, 0) rectangle ++(1.5,6.75);}
};
\node (FQ) at (6.9, 6) {
	\tikz {\draw[color = gray, rounded corners] (0, 0) rectangle ++(1.5,6.75);}
};

\draw ($(fx)+(-1, -.8)$) .. controls ($(fx)+(0, -.75)$) and ($(fx)+(1.2, -.25)$) .. ($(fx)+(1.5, 1.5)$);

\draw [gray, dotted, thick] (fy) to [bend right=25] ($(FX.north)+(-.75, -1)$);
\draw [gray, dotted, thick] (fx) to [bend left=25] ($(FY.north)+(-.5, -0.13)$);
\draw [gray, dotted, thick] (fz) to [bend left=35] ($(FZ.north)+(-.5, -0.13)$);
\draw [gray, dotted, thick] (fq) to [bend left=25] ($(FQ.north)+(-.75, -.5)$);

\node[draw=none,rotate=45] at ($(fx)+(-.5, .7)$) {\scriptsize clique};

\draw [thick, cyan] (c) to [bend right=25] (x1);
\node[draw=none,rotate=35,cyan] at (\x+1.1, \y+7.25) {$4d+c_z-c_x$};
\draw [thick, teal] (c) to [bend right=20] (y1);
\node[draw=none,rotate=35,teal] at (\x+3, \y+7.25) {$4d+c_z-c_y$};
\draw [thick, violet] (c) to [bend right=15] (z1);
\node[draw=none,rotate=35,violet] at (\x+4.2, \y+7) {$4d$};
\draw [thick, brown] (c) to [bend left=25] (q1);
\node[draw=none,rotate=-50,brown] at (\x+7, \y+7) {$3d$};

\draw [thick, violet] (c') to [bend left=25] (xd);
\node[draw=none,rotate=-35,violet] at (\x+.7, \y-1) {$4d$};
\draw [thick, orange] (c') to [bend left=20] (yd);
\node[draw=none,rotate=-35,orange] at (\x+2.95, \y-1.25) {$4d+c_y-c_x$};
\draw [thick, cyan] (c') to [bend left=15] (zd);
\node[draw=none,rotate=-40,cyan] at (\x+4.3, \y-1.15) {$4d+c_z-c_x$};
\draw [thick, brown] (c') to [bend left=15] (q1);
\node[draw=none,rotate=-50,brown] at (\x+5.15, \y-1) {$3d$};

\draw ($(fx)+(-.8, -2.8)$) .. controls ($(fx)+(.5, -.5)$) and ($(fx)+(1.5, 0.6)$) .. ($(fx)+(5, 1.3)$);
\end{tikzpicture}
		}
	\caption{Colorful edges denote paths of the depicted lengths. Two black curves separate the clique from the rest of the graph. Each dotted edge shows that a vertex of the clique is adjacent to all the vertices of the graph below the separating curves except those in the adjacent rectangle. The gray curve gives an intuition of which vertices of $P_x$, $P_y$, and $P_z$ are contained in $B_{r_c}(c)$ and $B_{r'_{c'}}(c')$.}\label{fig:whard}
\end{figure}

\begin{itemize}
	\item For each variable $x\in \mathcal{X}$, make a (variable) path $P^x:=(x_1,\ldots,x_d)$ of order~$d$.
	
	\item For each clause $c\in \mathcal{C}$, make two vertices $c$ and $c'$ and, w.l.o.g., suppose that the clause~$c$ contains the variables $x$, $y$, and $z$, and that $c_x\leq c_y\leq c_z$.
	Connect the vertex $c$ to $x_d$, $y_d$, and $z_d$ by three distinct paths of lengths~$4d$, $4d+c_y-c_x$, and $4d+c_z-c_x$, respectively.\footnote{``Connect two vertices $u,v$ by a path of length $p$'' means to make $u$ and $v$ the endpoints of a path of length~$p$.}
	These long paths ensure that there are no unwanted shortcuts in $G$ and that, for $r_c=5d-c_x-1$, the ball $B_{r_c}(c)$ contains all of $P^x$, $P^y$, and $P^z$ except for their first~$c_x$, $c_y$, and $c_z$ vertices (whose respective indices correspond to the respective variable values that satisfy the clause~$c$), respectively.
	Similarly, connect the vertex $c'$ to $x_1$, $y_1$, and $z_1$ by three distinct paths of lengths~$4d+c_z-c_x$, $4d+c_z-c_y$, and $4d$, respectively.
	Analogously, this ensures that, for $r'_{c'}=4d+c_z-1$, the ball $B_{r'_{c'}}(c')$ contains the opposite vertices to $B_{r_c}(c)$ in $P^x$, $P^y$, and $P^z$ (whose respective indices correspond to the respective variable values that do not satisfy the clause $c$), while ensuring that no unwanted shortcuts exist in~$G$. 
	
	\item For each clause $c\in \mathcal{C}$ and each variable $q\in \mathcal{X}$ such that $c$ does not contain $q$, in $G$, connect the vertices $c$ and $c'$ to $q_1$ by distinct paths of length~$3d$.
	These paths ensure that the balls $B_{r_c}(c)$ and $B_{r'_{c'}}(c')$ described above contain every other variable path completely, while ensuring that no unwanted shortcuts exist in $G$. 
	
	\item Let $S$ be the set of all the vertices that currently exist in $G$.
	For each $x\in \mathcal{X}$, in~$G$, make a vertex~$f_x$ and connect it to each vertex in~$S$ except those in~$P^x$ via a distinct path of length~$6d$.
	Finally, for all $x,y\in \mathcal{X}$, make $f_x$ adjacent to $f_y$.
	This ensures that, for each $x\in \mathcal{X}$, the ball $B_{6d}(f_x)$ contains every vertex in $G$ except for those in $P^x$, while ensuring that no unwanted shortcuts exist in $G$.
	This completes the construction of $G$ (see Figure~\ref{fig:whard}). 
\end{itemize}

Set $k:=|\mathcal{X}|$. Let $\mathcal{B}$ consist of $V(G)$, $B_{6d}(f_x)$ for all $x\in \mathcal{X}$, and $B_{r_c}(c)$ and $B_{r'_{c'}}(c')$ for all $c\in \mathcal{C}$.

\noindent\textit{Correctness of the Reduction.}
Suppose, given an instance $\phi$ of {\sc NAE-Integer-3-SAT}, that the reduction
from the subsection above returns $(G,\mathcal{B}, k)$ as an instance of \normalp.
\begin{lemma}\label{lem:w-correctness}
	$\phi$ is a YES-instance of {\sc NAE-Integer-3-SAT} if and only if $(G,\mathcal{B},k)$ is a YES-instance of \normalp. 
\end{lemma}

\begin{proof}
First, suppose that $\phi$ is a YES-instance of {\sc NAE-Integer-3-SAT}. We construct a positive teaching map $T$ as follows.
In the NAE-satisfying variable assignment for~$\phi$, for each variable $x\in \mathcal{X}$, if the integer~$j$ is assigned to $x$, then place $x_j$ in $T(V(G))$.
For each $c\in \mathcal{C}$, set $T(B_{r_c}(c)):=\{c\}$ and $T(B_{r'_{c'}}(c')):=\{c'\}$.
Lastly, for each $x\in X$, $T(B_{6d}(f_x))$ contains $f_x$ and one arbitrary vertex from each of the $k-1$ variable paths it contains.
We prove that $T$ is non-clashing for $\mathcal{B}$.

For any two balls $B_1,B_2\in \mathcal{B}$ centered at vertices of the form $c$ or $c'$ for the same or different clauses, $T$ satisfies the non-clashing condition since $T(B_1)$ contains its center (some $c$ or $c'$) while $B_2$ does not contain this vertex.
Indeed, the radius of $B_2$ is less than $5d$, while the distance between any two vertices of the form $c$ or $c'$ is at least $6d$ since any shortest path between them contains a vertex from a variable path. 
For any $x,y\in \mathcal{X}$, $T$ satisfies the non-clashing condition for $B_{6d}(f_x)$ and $B_{6d}(f_y)$ since $T(B_{6d}(f_x))$ contains a vertex in $P^y$ while $B_{6d}(f_y)$ does not.
For any $x\in \mathcal{X}$ and $c\in \mathcal{C}$, $T$ satisfies the non-clashing condition for $B_{6d}(f_x)$ and $B_{r_c}(c)$, as well as $B_{6d}(f_x)$ and $B_{r_{c'}}(c')$, since $T(B_{6d}(f_x))$ contains $f_x$ while $B_{r_c}(c)$ and $B_{r_{c'}}(c')$ do not as $r_c,r_{c'}<5d$ while $c$ and $c'$ are at distance $6d$ from $f_x$.
For any $x\in \mathcal{X}$, $T$ satisfies the non-clashing condition for $V(G)$ and $B_{6d}(f_x)$ since $T(V(G))$ contains a vertex in $P^x$ while $B_{6d}(f_x)$ does not.
Finally, for any $c\in \mathcal{C}$, $T$ satisfies the non-clashing condition for $V(G)$ and $B_{r_c}(c)$, as well as $V(G)$ and $B_{r_{c'}}(c')$, since, among the variables contained in the clause $c$, $T(V(G))$ contains at least one vertex from one of those variable paths whose index satisfies $c$, and at least one vertex from one of those variable paths whose index does not satisfy $c$.
As can be recalled from the construction, this implies that $B_{r_c}(c)$ and $B_{r_{c'}}(c')$ do not contain the respective vertices.
Thus, $T$ satisfies the non-clashing property for all pairs of balls in $\mathcal{B}$.

Now, we prove the reverse direction, so suppose that $(G,\mathcal{B},k)$ is a YES-instance of \normalp\ and that $T$ is the corresponding teaching map.
For all $x\in \mathcal{X}$, in order for $T$ to satisfy the non-clashing condition for $B_{6d}(f_x)$ and $V(G)$, we have that $T(V(G))$ contains at least one vertex from $P^x$ as $B_{6d}(f_x)\subset V(G)$ and $V(G)\setminus B_{6d}(f_x)$ is restricted to the vertices in $P^x$.
Since $k=|\mathcal{X}|$, we in fact have that $T(V(G))$ contains exactly one vertex from $P^x$ for all $x\in \mathcal{X}$.
Extract a variable assignment for $\phi$ from $T(V(G))$ as follows.
For each $x\in X$, assign the variable $x$ the value of the index of the unique vertex contained in both $P^x$ and $T(V(G))$.
We prove that this is an NAE-satisfying variable assignment for~$\phi$.

W.l.o.g., let $c\in \mathcal{C}$ be a clause containing the variables $x,y,z\in \mathcal{X}$.
In order for $T$ to satisfy the non-clashing condition for $V(G)$ and $B_{r_c}(c)$, $T(V(G))$ must contain at least one vertex in $P^x$, $P^y$ or $P^z$ that is not contained in $B_{r_c}(c)$.
Analogously, in order for $T$ to satisfy the non-clashing condition for $V(G)$ and $B_{r'_{c'}}(c')$, $T(V(G))$ must contain at least one vertex in $P^x$, $P^y$ or $P^z$ that is not contained in $B_{r'_{c'}}(c')$.
Recall that all of the vertices in $P^x$, $P^y$, and $P^z$ that are not contained in $B_{r_c}(c)$ have respective indices that correspond to the respective variable values that satisfy the clause $c$.
Similarly, recall that all of the vertices in $P^x$, $P^y$, and $P^z$ that are not contained in $B_{r_{c'}}(c')$ have respective indices that correspond to the respective variable values that do not satisfy the clause $c$.
As these arguments hold for any clause $c\in \mathcal{C}$, the variable assignment extracted above corresponds to an NAE-satisfying variable assignment for~$\phi$. 
\end{proof}

Now, we are ready to proceed with the proof of Theorem~\ref{thm:hard-fvs}.

\begin{proof}[Proof of Theorem~\ref{thm:hard-fvs}]
	Lemma~\ref{lem:w-correctness} establishes the correctness of the polynomial-time reduction from the beginning of Section~\ref{sec:whard}.
	To complete the proof, it remains to show that $\fvs(G)+\pw(G)+k$ is bounded above by a function of $|\mathcal{X}|$.
	This clearly holds for $k$ which is, by definition, $|\mathcal{X}|$.
	Deleting from $G$ the vertices $x_1$, $x_d$, and $f_x$ for all $x\in \mathcal{X}$ results in an acyclic graph $G'$; in particular $G$ has a feedback vertex set of size $3|\mathcal{X}|$.
	To establish a bound on $\pw(G)$, it now suffices to show that $G'$ also has bounded pathwidth.
	Note that $G'$ consists of a set of connected components, each of which is either a subdivided caterpillar (this is what remains of each component containing a variable path) or a vertex (of the form $c$ or $c'$) with multiple pendent subdivided caterpillars and (simple) paths. Since deleting one further vertex from each connected component may only reduce the pathwidth by $1$ and we need a single such deletion operation to reach a graph class of constant pathwidth (see Section~\ref{sec:prelims}), we also obtain that $\pw(G)$ is bounded by a function of~$|\mathcal{X}|$.
\end{proof}

\section{Concluding Remarks}
Our computational upper and lower bounds provide a near-comprehensive understanding of the complexity of computing the positive non-clashing teaching dimension. Apart from our contributions to the previously studied strict setting, we consider it notable that our work is the first to also tackle the complexity of non-clashing teaching in the non-strict setting---\emph{i.e.}, the more general (and arguably more natural) case where not all possible concepts are present.

One open question highlighted by our work concerns the tiny remaining gap between the algorithmic lower and upper bounds obtained in Theorem~\ref{thm:ETHlower} and Proposition~\ref{pro:exactalgo}. In particular, is there a way to improve the running time of the latter algorithm to $2^{\bigoh(|V(G)|\cdot d\cdot k)}$ and make the bounds tight? More general directions for future work are to perform a similar complexity analysis in the non-positive setting and to consider approximation algorithms.

\section*{Acknowledgements}
This work was funded by the Austrian Science Fund (FWF) [10.55776/Y1329 and 10.55776/COE12], the WWTF Vienna Science and Technology Fund (Project 10.47379/ICT22029), the \includegraphics[width=0.5cm]{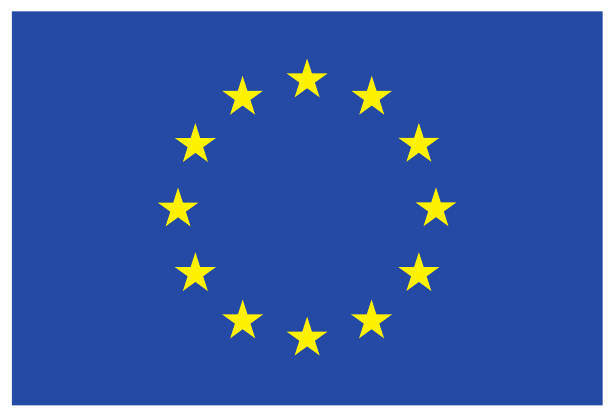} European Union's Horizon 2020 research and innovation COFUND programme LogiCS@TUWien (grant agreement No 101034440),
the Spanish Ministry of Economic Affairs and Digital Transformation and the European Union-NextGenerationEU through the  project 6G-RIEMANN~(TSI-063000-2021-147), and the Smart Networks and Services Joint Undertaking (SNS JU) under the European Union's Horizon Europe and innovation programme under Grant Agreement No. 101139067 (ELASTIC). Views and opinions expressed are
however those of the authors only and do not necessarily reflect those of the European
Union (EU). Neither the EU nor the granting authority can be held responsible for them.

\bibliography{NCTD_param_bib}

\newcommand{\etalchar}[1]{$^{#1}$}
\begin{thebibliography}{BHMvL16}

\bibitem[ACYT12]{ACYT12}
Baris Akgun, Maya Cakmak, Jae~W. Yoo, and Andrea~L. Thomaz.
\newblock Trajectories and keyframes for kinesthetic teaching: A human-robot
  interaction perspective.
\newblock In {\em Proc. of the 7th {ACM/IEEE} International Conference on
  Human-Robot Interaction ({HRI}~2012)}, pages 391--398, 2012.

\bibitem[Ang80a]{Angluinjcss}
Dana Angluin.
\newblock Finding patterns common to a set of strings.
\newblock {\em Journal of Computer and System Sciences}, 21:46--62, 1980.

\bibitem[Ang80b]{Angluininf}
Dana Angluin.
\newblock Inductive inference of formal languages from positive data.
\newblock {\em Information and Control}, 45:117--135, 1980.

\bibitem[BDL98]{BDLi}
Shai Ben-David and Ami Litman.
\newblock Combinatorial variability of {V}apnik-{C}hervonenkis classes with
  applications to sample compression schemes.
\newblock {\em Discrete Applied Mathematics}, 86:3--25, 1998.

\bibitem[BGS23]{BGS23}
Cornelius Brand, Robert Ganian, and Kirill Simonov.
\newblock A parameterized theory of {PAC} learning.
\newblock In {\em Proc. of the 37th {AAAI} Conference on Artificial
  Intelligence ({AAAI}~2023)}, pages 6834--6841, 2023.

\bibitem[BHMvL16]{BHML16}
Karl Bringmann, Danny Hermelin, Matthias Mnich, and Erik~Jan van Leeuwen.
\newblock Parameterized complexity dichotomy for steiner multicut.
\newblock {\em Journal of Computer and System Sciences}, 82(6):1020--1043,
  2016.

\bibitem[BN19]{BN19}
Daniel~S. Brown and Scott Niekum.
\newblock Machine teaching for inverse reinforcement learning: Algorithms and
  applications.
\newblock In {\em Proc. of the 33rd {AAAI} Conference on Artificial
  Intelligence ({AAAI}~2019)}, volume~33, pages 7749--7758, 2019.

\bibitem[CCM{\etalchar{+}}23]{ChChMc}
J\'{e}r\'{e}mie Chalopin, Victor Chepoi, Fionn {Mc Inerney}, S\'{e}bastien
  Ratel, and Yann Vax\`{e}s.
\newblock Sample compression schemes for balls in graphs.
\newblock {\em {SIAM} Journal on Discrete Mathematics}, 37(4):2585--2616, 2023.

\bibitem[CCMR24]{ChalopinCIR24}
J{\'{e}}r{\'{e}}mie Chalopin, Victor Chepoi, Fionn {Mc~Inerney}, and
  S{\'{e}}bastien Ratel.
\newblock Non-clashing teaching maps for balls in graphs.
\newblock In {\em Proc. of the 37th Annual Conference on Learning Theory
  ({COLT} 2024)}, volume 247 of {\em Proceedings of Machine Learning Research},
  pages 840--875, 2024.

\bibitem[CCMW22]{ChChMoWa}
J\'{e}r\'{e}mie Chalopin, Victor Chepoi, Shay Moran, and Manfred~K. Warmuth.
\newblock Unlabeled sample compression schemes and corner peelings for ample
  and maximum classes.
\newblock {\em Journal of Computer and System Sciences}, 127:1--28, 2022.

\bibitem[CFK{\etalchar{+}}15]{CyganFKLMPPS15}
Marek Cygan, Fedor~V. Fomin, Lukasz Kowalik, Daniel Lokshtanov, D{\'{a}}niel
  Marx, Marcin Pilipczuk, Michal Pilipczuk, and Saket Saurabh.
\newblock {\em Parameterized Algorithms}.
\newblock Springer, 2015.

\bibitem[CSM{\etalchar{+}}18]{CAMPY18}
Yuxin Chen, Adish Singla, Oisin {Mac Aodha}, Pietro Perona, and Yisong Yue.
\newblock Understanding the role of adaptivity in machine teaching: The case of
  version space learners.
\newblock In {\em Advances in Neural Information Processing Systems
  ({NeurIPS}~2018)}, volume~31, pages 1483--1493, 2018.

\bibitem[DDvtH16]{DrangeDH16}
P{\aa}l~Gr{\o}n{\aa}s Drange, Markus~S. Dregi, and Pim van~'t Hof.
\newblock On the computational complexity of vertex integrity and component
  order connectivity.
\newblock {\em Algorithmica}, 76(4):1181--1202, 2016.

\bibitem[DEF93]{DF93}
Rodney~G. Downey, Patricia~A. Evans, and Michael~R. Fellows.
\newblock Parameterized learning complexity.
\newblock In {\em Proc. of the 6th Annual ACM Conference on Computational
  Learning Theory ({COLT}~1993)}, pages 51--57, 1993.

\bibitem[Den01]{Denis01}
Fran\c{c}ois Denis.
\newblock Learning regular languages from simple positive examples.
\newblock {\em Machine Learning}, 44:37--66, 2001.

\bibitem[DF13]{DowneyF13}
Rodney~G. Downey and Michael~R. Fellows.
\newblock {\em Fundamentals of Parameterized Complexity}.
\newblock Texts in Computer Science. Springer, 2013.

\bibitem[Die24]{Diestel}
Reinhard Diestel.
\newblock {\em Graph Theory, 6th Edition}, volume 173 of {\em Graduate texts in
  mathematics}.
\newblock Springer, 2024.

\bibitem[EGK{\etalchar{+}}23]{EibenGKOS23}
Eduard Eiben, Robert Ganian, Iyad~A. Kanj, Sebastian Ordyniak, and Stefan
  Szeider.
\newblock The computational complexity of concise hypersphere classification.
\newblock In {\em Proc. of the International Conference on Machine Learning
  ({ICML} 2023)}, volume 202 of {\em Proceedings of Machine Learning Research},
  pages 9060--9070, 2023.

\bibitem[EOPS23]{EibenOPS23}
Eduard Eiben, Sebastian Ordyniak, Giacomo Paesani, and Stefan Szeider.
\newblock Learning small decision trees with large domain.
\newblock In {\em Proc. of the 32nd International Joint Conference on
  Artificial Intelligence ({IJCAI} 2023)}, pages 3184--3192, 2023.

\bibitem[FKS{\etalchar{+}}23]{FKS23}
Shaun Fallat, David Kirkpatrick, Hans~U. Simon, Abolghasem Soltani, and Sandra
  Zilles.
\newblock On batch teaching without collusion.
\newblock {\em Journal of Machine Learning Research}, 24:40:1--40:33, 2023.

\bibitem[Flo89]{Fl}
S.~Floyd.
\newblock {\em Space-bounded learning and the {V}apnik-{C}hervonenkis
  dimension}.
\newblock PhD thesis, U.C. Berkeley, 1989.

\bibitem[FS89]{fellows1989immersion}
Michael~R Fellows and Sam Stueckle.
\newblock The immersion order, forbidden subgraphs and the complexity of
  network integrity.
\newblock {\em J. Combin. Math. Combin. Comput}, 6(1):23--32, 1989.

\bibitem[GHK{\etalchar{+}}22]{GimaHKKO22}
Tatsuya Gima, Tesshu Hanaka, Masashi Kiyomi, Yasuaki Kobayashi, and Yota
  Otachi.
\newblock Exploring the gap between treedepth and vertex cover through vertex
  integrity.
\newblock {\em Theoretical Computer Science}, 918:60--76, 2022.

\bibitem[GK95]{GK95}
Sally~A. Goldman and Michael~J. Kearns.
\newblock On the complexity of teaching.
\newblock {\em Journal of Computer and System Sciences}, 50(1):20--31, 1995.

\bibitem[GK21]{GK21}
Robert Ganian and Viktoriia Korchemna.
\newblock The complexity of {B}ayesian network learning: Revisiting the
  superstructure.
\newblock In {\em Advances in Neural Information Processing Systems
  ({NeurIPS}~2021)}, volume~34, pages 430--442, 2021.

\bibitem[GM96]{GM96}
Sally~A. Goldman and H.~David Mathias.
\newblock Teaching a smarter learner.
\newblock {\em Journal of Computer and System Sciences}, 52(2):255--267, 1996.

\bibitem[GO24]{GimaO24}
Tatsuya Gima and Yota Otachi.
\newblock Extended {MSO} model checking via small vertex integrity.
\newblock {\em Algorithmica}, 86(1):147--170, 2024.

\bibitem[GRSZ17]{GCS17}
Ziyuan Gao, Christoph Ries, Hans~U. Simon, and Sandra Zilles.
\newblock Preference-based teaching.
\newblock {\em Journal of Machine Learning Research}, 18:1--32, 2017.

\bibitem[HLM{\etalchar{+}}16]{HLMCA16}
Mark~K. Ho, Michael Littman, James {MacGlashan}, Fiery Cushman, and Joseph~L.
  Austerweil.
\newblock Showing versus doing: Teaching by demonstration.
\newblock In {\em Advances in Neural Information Processing Systems
  ({NeurIPS}~2016)}, volume~33, pages 3027--3035, 2016.

\bibitem[HLVY24]{HanakaLVY24}
Tesshu Hanaka, Michael Lampis, Manolis Vasilakis, and Kanae Yoshiwatari.
\newblock Parameterized vertex integrity revisited.
\newblock In {\em Proc. of the 49th International Symposium on Mathematical
  Foundations of Computer Science ({MFCS} 2024)}, volume 306 of {\em LIPIcs},
  pages 58:1--58:14. Schloss Dagstuhl - Leibniz-Zentrum f{\"{u}}r Informatik,
  2024.

\bibitem[IPZ01]{IPZ}
Russell Impagliazzo, Ramamohan Paturi, and Francis Zane.
\newblock Which problems have strongly exponential complexity?
\newblock {\em Journal of Computer and System Sciences}, 63(4):512--530, 2001.

\bibitem[KSZ19]{KSZ19}
David Kirkpatrick, Hans~U. Simon, and Sandra Zilles.
\newblock Optimal collusion-free teaching.
\newblock In {\em Proc. of the 30th International Conference on Algorithmic
  Learning Theory ({ALT}~2019)}, volume~98 of {\em Proceedings of Machine
  Learning Research}, pages 506--528, 2019.

\bibitem[KW07]{KuWa}
Dima Kuzmin and Manfred~K. Warmuth.
\newblock Unlabelled compression schemes for maximum classes.
\newblock {\em Journal of Machine Learning Research}, 8:2047--2081, 2007.

\bibitem[LL18]{LL18}
Yuanzhi Li and Yingyu Liang.
\newblock Learning mixtures of linear regressions with nearly optimal
  complexity.
\newblock In {\em Proc. of the 31st Conference on Learning Theory
  ({COLT}~2018)}, volume~75 of {\em Proceedings of Machine Learning Research},
  pages 1125--1144, 2018.

\bibitem[LW86]{LiWa}
N.~Littlestone and M.~K. Warmuth.
\newblock Relating data compression and learnability.
\newblock {\em Unpublished}, 1986.

\bibitem[MCV{\etalchar{+}}19]{MCV19}
Farnam Mansouri, Yuxin Chen, Ara Vartanian, Jerry Zhu, and Adish Singla.
\newblock Preference-based batch and sequential teaching: Towards a unified
  view of models.
\newblock In {\em Advances in Neural Information Processing Systems
  ({NeurIPS}~2019)}, volume~32, pages 9195--9205, 2019.

\bibitem[MZ15]{MZ15}
Shike Mei and Xiaojin Zhu.
\newblock Using machine teaching to identify optimal training-set attacks on
  machine learners.
\newblock In {\em Proc. of the 29th {AAAI} Conference on Artificial
  Intelligence ({AAAI}~2015)}, volume~29, pages 2871--2877, 2015.

\bibitem[OS21]{OrdyniakS21}
Sebastian Ordyniak and Stefan Szeider.
\newblock Parameterized complexity of small decision tree learning.
\newblock In {\em Proc. of the 35th {AAAI} Conference on Artificial
  Intelligence ({AAAI} 2021)}, pages 6454--6462, 2021.

\bibitem[Sim23]{Si}
Hans~U. Simon.
\newblock Tournaments, johnson graphs, and {NC}-teaching.
\newblock In {\em Proc. of the 34th International Conference on Algorithmic
  Learning Theory ({ALT}~2023)}, volume 201 of {\em Proceedings of Machine
  Learning Research}, pages 1411--1428, 2023.

\bibitem[SM91]{SM91}
Ayumi Shinohara and Satoru Miyano.
\newblock Teachability in computational learning.
\newblock {\em New Generation Computing}, 8:337--347, 1991.

\bibitem[SO94]{SO94}
Andreas Stolcke and Stephen Omohundro.
\newblock Inducing probabilistic grammars by {B}ayesian model merging.
\newblock In {\em Proc. of the 2nd International Colloquium on Grammatical
  Inference ({ICGI}~1994)}, pages 106--118, 1994.

\bibitem[SPK00]{SPK00}
Ingo Schwab, Wolfgang Pohl, and Ivan Koychev.
\newblock Learning to recommend from positive evidence.
\newblock In {\em Proc. of the 5th International Conference on Intelligent User
  Interfaces ({IUI}~2000)}, pages 241--247, 2000.

\bibitem[TC09]{TC09}
Andrea~L. Thomaz and Maya Cakmak.
\newblock Learning about objects with human teachers.
\newblock In {\em Proc. of the 4th {ACM/IEEE} International Conference on
  Human-Robot Interaction ({HRI}~2009)}, pages 15--22, 2009.

\bibitem[THF19]{TelleHF19}
Jan~Arne Telle, Jos{\'{e}} Hern{\'{a}}ndez{-}Orallo, and C{\`{e}}sar Ferri.
\newblock The teaching size: computable teachers and learners for universal
  languages.
\newblock {\em Machine Learning}, 108(8-9):1653--1675, 2019.

\bibitem[WDMH06]{WDM06}
Chunlin Wang, Chris Ding, Richard~F. Meraz, and Stephen~R. Holbrook.
\newblock {PSoL}: a positive sample only learning algorithm for finding
  non-coding {RNA} genes.
\newblock {\em Bioinformatics}, 22(21):2590--2596, 2006.

\bibitem[YJSS08]{YJSS08}
Malik Yousef, Segun Jung, Louise~C. Showe, and Michael~K. Showe.
\newblock Learning from positive examples when the negative class is
  undetermined - micro{RNA} gene identification.
\newblock {\em Algorithms for Molecular Biology}, 3(2), 2008.

\bibitem[Zhu15]{Zhu15}
Xiaojin Zhu.
\newblock Machine teaching: An inverse problem to machine learning and an
  approach toward optimal education.
\newblock In {\em Proc. of the 29th {AAAI} Conference on Artificial
  Intelligence ({AAAI}~2015)}, volume~29, pages 4083--4087, 2015.

\bibitem[ZLHZ11]{ZLH11}
Sandra Zilles, Steffen Lange, Robert Holte, and Martin Zinkevich.
\newblock Models of cooperative teaching and learning.
\newblock {\em Journal of Machine Learning Research}, 12:349--384, 2011.

\bibitem[ZSZR18]{ZSZR18}
Xiaojin Zhu, Adish Singla, Sandra Zilles, and Anna~N. Rafferty.
\newblock An overview of machine teaching.
\newblock {\em Arxiv:1801.05927}, 2018.

\bibitem[ZZW18]{ZZW18}
Xuezhou Zhang, Xaiojin Zhu, and Stephen Wright.
\newblock Training set debugging using trusted items.
\newblock In {\em Proc. of the 32nd {AAAI} Conference on Artificial
  Intelligence ({AAAI}~2018)}, pages 4482--4489, 2018.

\end{thebibliography}
\bibliographystyle{alpha}

\end{document}